\documentclass[aos,preprint]{imsart_v2}
\usepackage{amsthm,amsmath,natbib,graphicx,amssymb,float, amstext,bm,epstopdf,enumerate,xcolor,subcaption,multirow,xr}
\RequirePackage[colorlinks,citecolor=blue,urlcolor=blue]{hyperref}
\numberwithin{equation}{section}

\startlocaldefs

\newtheorem{Theorem}{Theorem}
\newtheorem{Lemma}{Lemma}
\newtheorem{Corollary}{Corollary}
\newtheorem{Proposition}{Proposition}
\newtheorem{Example}{Example}

\definecolor{DarkRed}{rgb}{.7,0,.4}

\def\ci{\cite}
\def\cp{\citep}
\def\mc{\mathcal}
\def\bco{\iffalse}

\newcommand{\be}{\begin{eqnarray}}
\newcommand{\ee}{\end{eqnarray}}

\newcommand{\R}{\mathcal{R}}
\newcommand{\ip}[2]{\langle #1, #2 \rangle}
\def\sn{\sum_{i=1}^n}
\def\eps{\varepsilon}

\def\om{\omega}
\def\mp{m_{\oplus}}
\def\hm{\hat{m}}
\def\hl{\hat{l}}

\def\tl{\tilde{l}}

\def\hmp{\hm_{\oplus}}
\def\tlp{\tl_{\oplus}}
\def\hlp{\hl_{\oplus}}
\def\omp{\om_{\oplus}}
\def\homp{\hat{\om}_\oplus}
\def\hbeta{\hat{\beta}}

\def\xbar{\bar{X}}
\def\s2{\sigma^2}
\def\hs2{\hat{\sigma}^2}
\def\syx{\sigma_{YX}}
\def\Si{\Sigma}
\def\hSi{\hat{\Sigma}}
\def\Var{{\rm Var}}
\def\Cov{{\rm Cov}}
\def\ra{\rightarrow}
\DeclareMathOperator*{\argmax}{argmax}
\DeclareMathOperator*{\argmin}{argmin}

\def\Om{\Omega}
\def\om{\omega}

\newcommand{\bn}{\beta_0}
\newcommand{\hbn}{\hat{\beta}_0}
\newcommand{\bo}{\beta_1}
\newcommand{\hbo}{\hat{\beta}_1}
\def\sbo{{\beta_1}^{\ast}}

\def\sbn{\beta_0^{\ast}}
\def\dyx{{\rm d}F_{Y|X}(x, y)}
\def\dx{{\rm d}F_X(x)}
\def\dF{{\rm d}F(x, y)}
\def\dFz{{\rm d}F(z, y)}

\def\dy{{\rm d}F_Y(y)}
\def\inv{^{-1}}
\def\Ltwoip#1#2{\langle #1, #2\rangle_{L^2}}
\def\Ltwonorm#1{\lVert #1 \rVert_{L^2}}
\def\Ltwo#1#2{d_{L^2}(#1, #2)}
\def\Froip#1#2{\langle #1, #2\rangle_{F}}
\def\Fronorm#1{\lVert #1 \rVert_F}
\def\Fro#1#2{d_F(#1,#2)}

\bibpunct{(}{)}{;}{a}{}{,}
\def\hmu{\hat{\mu}}

\def\tL{\tilde{L}_n}
\def\hL{\hat{L}_n}

\def\diam{\textrm{diam}}
\def\ed{\end{document}}
\def\vecop{{\rm vec}}
\def\vechop{{\rm vech}}

\def\dw{d_W}

\endlocaldefs

\begin{document}

\renewcommand{\MR}[1]{#1}

\begin{frontmatter}

\title{Fr\'echet Regression For  Random Objects With Euclidean Predictors}
\runtitle{Fr\'echet Regression}


\begin{aug}
\author{\fnms{Alexander} \snm{Petersen}\thanksref{m1,t2}\corref{}\ead[label=e1]{petersen@stat.byu.edu}}
 \and
\author{\fnms{Hans-Georg} \snm{M\"uller}\thanksref{m2,t1}\ead[label=e2]{hgmueller@ucdavis.edu}}
\affiliation{Department of Statistics, University of California, Santa Barbara\thanksmark{m1}}
\affiliation{Department of Statistics, University of California, Davis\thanksmark{m2}}
\thankstext{t1}{Supported in part by National Science Foundation grants DMS-12-28369 and DMS-14-07852}
\thankstext{t2}{Present address: Department of Statistics, Brigham Young University}
\address{Address of Alexander Petersen\\
Department of Statistics \\
2152 WVB \\
Brigham Young University \\
Provo, UT 84602 \\
\printead{e1}}
\address{Address of Hans-Georg M\"uller\\
Department of Statistics \\
Mathematical Sciences Building 4118 \\
399 Crocker Lane \\
University of California, Davis \\
One Shields Avenue \\
Davis, CA 95616 \\
\printead{e2}}

\runauthor{Petersen and M\"uller}

\end{aug}

\begin{abstract}
Increasingly, statisticians are faced with the task of analyzing complex data that are non-Euclidean and specifically do not lie in a vector space.  To address the need for statistical methods for such data, we introduce the concept of Fr\'echet regression. This is a general approach to regression when responses are complex random objects in a metric space and predictors are in $\mathcal{R}^p$, achieved by extending the classical concept of a Fr\'echet mean to the notion of a conditional Fr\'echet mean. We develop generalized versions of both global least squares regression and local weighted least squares smoothing. The target quantities are appropriately defined population versions of global and local regression for response objects in a metric space.  We derive asymptotic rates of convergence for the corresponding fitted regressions using observed data to the population targets under suitable regularity conditions by applying empirical process methods. For the special case of random objects that reside in a Hilbert space, such as regression models with vector predictors and functional data as responses, we obtain a limit distribution. The proposed methods have broad applicability.  Illustrative examples include responses that consist of probability distributions and correlation matrices, and we demonstrate both global and local Fr\'echet regression for demographic and brain imaging data.  Local Fr\'echet regression is also illustrated via a simulation with response data which lie on the sphere.
\end{abstract}

\begin{keyword}[class=MSC]
\kwd[Primary ]{62G05}
\kwd[; secondary ]{62J99, 62G08}
\end{keyword}

\begin{keyword}
\kwd{Least Squares Regression; Random Objects; Metric Spaces; Local Linear Regression; Functional Connectivity; Densities as Objects}
\end{keyword}

\end{frontmatter}

\section{Introduction}
\label{sec: intro}

The regression relationship between a response variable and one or more predictor variables constitutes the target of many statistical methodologies.  The most basic form is linear regression, where all variables are real-valued, and the conditional mean of the response variable is linear in the predictors.  The linear regression model is quite flexible, includes polynomial fits and categorical predictor variables, among others, and remains one of the most popular  tools for data analysis.  In addition to the superb interpretability of linear models and simple model fitting via least squares, powerful inferential methods, with well-established theory, are available for estimation and testing.  Linear regression ideas also motivate local polynomial smoothing, further adding to their vast applicability.

In recent years, as data types are becoming  more complex, attention has turned to regression in more abstract settings. The importance of the analysis of such object data has recently been highlighted  \cp{marr:14,wang:07:1}. A setting that is increasingly encountered is that of a response variable taking values in a metric space, which may or may not have algebraic structure.  The presence of a metric provides a natural connection to the work of \ci{frec:48}, where the Fr\'echet mean is defined for random elements of a metric space as a direct generalization of the standard mean, which is defined by integration over a probability space.  This generalization has been increasingly exploited in statistical analyses due to its inherent flexibility. Specifically, no ambient vector space needs to be assumed and only a distance between data objects is required. As regression can be viewed as the modeling of conditional means, a key feature of our approach is that we introduce the concept of a conditional Fr\'echet mean, generalizing the classical Fr\'echet mean.

One important class of random objects, which has been extensively studied, consists of observations on a finite-dimensional differentiable Riemannian manifold.  Due to local Euclidean properties of the space, one can mimic both parametric (global) and nonparametric (local) regression techniques for standard Euclidean data quite effectively by local Euclidean approximations. Regression models for this special case have been well studied \cp{fish:87, chan:89, pren:89, fish:95}, including intrinsic models for geodesic regression \cp{flet:13,niet:11,corn:17}, semiparametric regression \cp{shi:09} and local kernel regression as a generalization of the classical Nadaraya-Watson smoother \cp{pell:06,davi:07,hink:12,yuan:12}.  Recently, the extrinsic regression model in \ci{lin:15} extends the notion of extrinsic means \cite[see, e.g., Ch. 11 and 18 of ][]{patr:15}, where  extrinsic approaches have been reported to have  computational advantages \cp{bhat:12}.

In this paper, however, we go beyond manifolds and our focus is on a more general case of random objects in metric spaces with little structure, where only distances between response objects are computable.  To our knowledge, in general metric spaces, the only global or parametric model which has been proposed is that of \ci{fara:14}, where data are represented as scores in a Euclidean space based on their pairwise distances, followed by the use of classical regression techniques.  This method requires a complicated ``backscoring" step, where vectors in Euclidean space are then represented in the original metric space, and its theoretical properties have not been studied.   Local regression methods on generic metric spaces are limited to Nadaraya-Watson type estimators \cp{davi:07,hein:09,stei:09,stei:10} and lack a comprehensive asymptotic analysis. Thus, there is a need for additional statistical models to tackle this type of data that is increasingly common. Accordingly, we present here methodology and theory for both global and local regression analysis of complex random objects.

Specifically, we consider regression relationships between responses which are complex random objects and vectors of real-valued predictors.  To this end, we develop a global regression relation as a generalization of multiple linear regression, as well as a class of more flexible local regression methods that generalizes local linear or polynomial regression. As the proposed regression approach for random objects incorporates the geometry implied by the metric and can be viewed as an extension of the Fr\'echet mean, we refer to our methods as Fr\'echet regression.  {\it Global Fr\'echet regression} provides an improvement on the global method of \ci{fara:14}, as the proposed model defines the regression directly on the object space and does not require backscoring.  The global Fr\'echet regression model constitutes a class of regression functions on arbitrary metric spaces which can be fitted without a tuning parameter or the need for any local smoothing technique.  We also propose \emph{local Fr\'echet regression}, which generalizes local linear estimation to a framework where responses are random objects, extending the available nonparametric regression methodology for object data. A challenge for the development of  local Fr\'echet regression is to define an appropriate population model, which serves as the target to which the fitted local Fr\'echet regression converges. We establish consistency and rates of convergence for both global and local Fr\'echet regression.

The proposed global Fr\'echet regression model is introduced in Section~\ref{sec: regress}, and theory quantifying the convergence rates of these estimators is given in Section~\ref{sec: theory}, along with some concrete examples which are shown to satisfy the necessary regularity conditions.  Local  Fr\'echet regression is introduced in 
Section~\ref{sec: local}, along with  asymptotic convergence theory. 
All proofs can be found in the Appendix. 
For the special case where the random objects take values in a Hilbert space, a limiting distribution can be obtained, as demonstrated in Section~\ref{sec: hilbert}.   

Our primary application examples deal with samples of probability distributions and correlation matrices, which are illustrated with data from demography and neuroimaging, with details in Sections~\ref{sec: density} and \ref{sec: correlation}, respectively. Here, we also include a discussion of practical issues, such as a suitable notion of the coefficient of determination $R^2$ when the responses are random objects.  For the space of probability distributions, we utilize the Wasserstein metric to conduct a simulation experiment as well as analyze the evolution of mortality profiles for two countries.  For the case where responses are correlation matrices, we examine the relationship between functional connectivity in the brain, as quantified by pairwise correlations of fMRI signals, with age as predictor.  Lastly, although the proposed methodology does not require any particular metric structure, it is nevertheless applicable to structured spaces such as manifolds.  To demonstrate this, the local Fr\'echet regression technique is also illustrated with simulated manifold data on the sphere $S^2 \subset \R^3$ in Section~\ref{sec: sphere}.

\section{Global Fr\'echet Regression}
\label{sec: regress}

\subsection{Preliminaries}
\label{ss: prelim}

Let $(\Om, d)$ be a metric space.  We consider a random process $(X, Y) \sim F$, where $X$ and $Y$ take values in $\R^p$ and $\Om$, respectively, and $F$ is the joint distribution of $(X, Y)$ on $\R^p \times \Om$. We denote the marginal distributions of $X$ and $Y$ as $F_X$ and $F_Y$, respectively, and assume that $E\lVert X\rVert_E^4 < \infty$, where $\lVert \cdot \rVert_E$ is the Euclidean norm on $\R^p$, so that $\mu = E(X)$ and $\Sigma = \Var(X)$ exist, with $\Sigma$ positive definite. The conditional distributions $F_{X|Y}$ and $F_{Y|X}$ are also assumed to exist.  In this general setting, we refer to $Y$ as a random object.  The usual notions of mean and variance were generalized to random objects in metric spaces in \ci{frec:48}, where
\begin{equation}
\label{eq: frechet_mv}
\omp = \argmin_{\om \in \Om} E(d^2(Y, \om)), \quad V_\oplus = E(d^2(Y, \omp))
\end{equation}
were defined, now commonly referred to as Fr\'echet mean and Fr\'echet variance, respectively.

Building on these concepts, we introduce the Fr\'echet regression function of $Y$ given $X = x$,
\begin{equation}
\label{eq: frechet_regression}
\mp(x) = \argmin_{\om \in \Om} M_\oplus(\om, x), \quad M_\oplus(\cdot, x) = E(d^2(Y, \cdot)|X = x),
\end{equation}
where we refer to $M_\oplus(\cdot, x)$ as  the (conditional) Fr\'echet function.  For the special case $\Om = \R$, various nonparametric regression methods have been developed which are based on kernel or local linear polynomial fitting \cp{fan:96}, splines \cp{crav:79,marx:96} or other smoothers.

A basic statistical task is to fit a global regression model for response $Y$ and predictor $X$, in order to provide ease of implementation and interpretation and allow for good options for overall inference and testing.  Fitting of such a global model also does not require the choice of a tuning parameter, as all local fitting methods do, since global models are usually fitted under the assumption that there is no bias.  Given that no algebraic structure is assumed, it is not feasible to directly generalize parametric models to a parametric function on $\Om$, as has been done in the special case when $\Om$ is a Riemannian manifold.  However, an alternative solution that we will develop is to recharacterize the standard multiple linear regression model as a function of weighted Fr\'echet means, where the weights have a known form and vary with $x$.

\subsection{Generalizing Linear Regression}
\label{ss: frechet}

We begin by considering the standard setup for linear regression, for which $\Om = \R$, and then write $m = m_\oplus$ in (\ref{eq: frechet_regression}).  The model for linear regression is
\begin{equation}
\label{eq: lin}
m(x) := E(Y|X = x) = \sbn + (\sbo)^T(x - \mu),
\end{equation}
where the scalar intercept $\sbn$ and slope vector $\sbo$ are the solutions
\begin{equation}
\label{eq: lin_ls}
(\beta_0^{\ast},\beta_1^{\ast})=\argmin_{\beta_0\in \R, \beta_1\in\R^p} \int\left[\int y\dyx-(\beta_0 +\beta_1^T(x - \mu))\right]^2\dx.
\end{equation}
Similar to the Fr\'echet mean, the goal is to characterize the regression values in (\ref{eq: lin}) as minimizers of weighted least squares problems, where the weights depend on predictor values and the squared distances depend on response values.  
Setting $\mu = E(X)$, $\Si = \Var(X)$ and $\syx = E\left[Y(X - \mu)\right]$, the normal equations for the right-hand side of (\ref{eq: lin_ls}) lead to
\[
E(Y) - \bn =0, \quad  \syx - \Sigma \bo=0,
\]
with solutions $\sbo=\Si\inv\syx$ and $\sbn = E(Y).$  Plugging these into (\ref{eq: lin}),
\begin{align}
m(x) &=E(Y) + \syx^T\Si\inv(x - \mu) = \int y\left\{1 + (z - \mu)^T\Si\inv(x - \mu)\right\}\dFz  \label{eq: int_weight}\\
&= \int y s(z, x) {\rm d}F(z, y), \nonumber
\end{align}
where the weight function $s$ is
\begin{equation}
\label{eq: weight}
s(z, x)= 1 + (z - \mu)^T\Si\inv(x - \mu).
\end{equation}
Because $\int s(z, x)\dFz = 1$, the last line of (\ref{eq: int_weight}) reveals that the standard linear regression function value $m(x)$ is the solution
\begin{equation}
\label{eq: regress}
m(x) = \argmin_{y \in \R} E\left[s(X, x)d_E^2(Y, y)\right],
\end{equation}
where $d_E$ is the standard Euclidean metric.  This alternative formulation of the linear regression function provides the key to defining the proposed global Fr\'echet regression function $\mp$ on an arbitrary metric space $(\Om, d)$, by simply replacing the Euclidean metric $d_E$, which is the default metric for real valued responses, by a more general metric  $d$ that is suitable for responses in  $\Om$.   The global Fr\'echet regression model then becomes
\begin{equation}
\label{eq: mp}
\mp(x):= \argmin_{\om \in \Om} M(\om, x), \quad M(\cdot, x) = E\left[s(X, x)d^2(Y, \cdot)\right].\end{equation}

Hence, generalizing multiple linear regression to the case of a metric-valued response is achieved by viewing the regression function as a sequence of weighted Fr\'echet means, with weights that are derived from those of the corresponding standard linear regression.  Although $\Omega$ is not a linear space, the weight function $s$ is a sensible choice for a number of reasons.  First, any coherent generalization of multiple linear regression to a global model for random object regression should result in a regression function passing through the point $(\mu, \omega_\oplus)$, which holds for the proposed model since $s(\cdot, \mu) \equiv 1$ implies that $m_\oplus(\mu) = \omega_\oplus$.  Second, in contrast to local regression in metric spaces, where the weights are given by a nonnegative kernel function, the weights given by $s$ can be negative and do not go to zero away from $x$, both of which are natural properties of a global regression relationship.  Lastly, despite being defined as a minimizer of a weighted Fr\'echet function, the proposed global Fr\'echet regression function can be computed analytically in some cases, in addition to the obvious case $\Omega = \mathcal{R}$.  As an illustrative example, when $\Om$ is the space of probability distributions on the real line equipped with the Wasserstein metric (see Section~\ref{ss: wass_sim} below) and the random objects $Y$ are distributions from a location-scale family with random location $\nu$ and scale $\sigma$, the global Fr\'echet regression model is equivalent to modeling the conditional means of $\nu$ and $\sigma$ as linear functions of the predictor $x.$   In fact, when the location-scale family is the Gaussian family, this space has a curved manifold structure, with properties studied extensively in the literature (e.g. Takatsu, 2011).  This provides an example of a curved manifold for which the global Fr\'echet regression relationship is sensible.

\subsection{Estimation}
\label{ss: estimation}

Assume that $(X_i, Y_i) \sim F$, $i = 1,\ldots,n$, are independent.  We take the standard approach to estimate the minimizer in (\ref{eq: mp}) by substituting the empirical distribution ${\rm d}F_n$ for ${\rm d}F$ in the integral in (\ref{eq: mp}).  Additionally, the unknown parameters $\mu$ and $\Sigma$ in (\ref{eq: weight}) are replaced by their empirical estimates $\xbar = n\inv\sn X_i$ and $\hSi = n\inv\sn(X_i - \xbar)(X_i - \xbar)^T$, respectively.

The empirical weights
\begin{equation}
\label{eq: emp_weights}
s_{in}(x) := 1 + (X_i - \xbar)^T\hSi\inv(x - \xbar)
\end{equation}
then lead to the estimator
\begin{equation}
\label{eq: mp_est}
\hmp(x) = \argmin_{\om \in \Om} M_n(\om, x)
\end{equation}
of $\mp(x)$ for $x \in \R^p$,
where $M_n(\cdot, x) = n\inv\sn s_{in}(x)d^2(Y_i, \om).$

\section{Theory}
\label{sec: theory}

We first consider the estimation of the regression relation in (\ref{eq: mp}) by the corresponding estimator in (\ref{eq: mp_est}) in the case of a totally bounded metric space $(\Om, d)$.  Recall the functions
\[
M(\om, x) := E\left[s(X, x)d^2(Y, \om)\right], \quad M_n(\om, x) := n\inv \sn s_{in}(x)d^2(Y_i, \om).
\]
With regard to the objects in (\ref{eq: mp}) and (\ref{eq: mp_est}), we require the following assumptions for a fixed $x \in \R^p$.
\begin{itemize}
  \item[(P0)] The objects $\mp(x)$ and $\hmp(x)$ exist and are unique, the latter almost surely, and, for any $\eps > 0$, \mbox{$\inf_{d(\om, \mp(x)) > \eps} M(\om, x) > M(\mp(x), x)$.}
  \item[(P1)] 
  Let $B_\delta(m_\oplus(x)) \subset \Omega$ be the ball of radius $\delta$ centered at $m_\oplus(x)$ and $N(\epsilon, B_\delta(m_\oplus(x)), d)$ be its covering number using balls of radius $\epsilon.$ There exists $\gamma \geq 0$, possibly depending on $x$, such that
  $$
  \int_0^1 \sqrt{1 + \log N(\delta\epsilon, B_\delta(\mp(x)), d)} \;d\epsilon = O(\delta^{-\gamma}) \quad \textrm{as} \quad \delta \rightarrow 0.
  $$
  \item[(P2)] There exist $\eta > 0$, $C > 0$ and $\beta > 1$, possibly depending on $x$, such that, whenever $d(\mp(x),\om) < \eta$, we have $M(\om, x) - M(\mp(x), x) \geq Cd(\om, \mp(x))^\beta$.
\end{itemize}
Assumption (P0) is common to establish the consistency of an $M$-estimator such as $\hmp(x)$; see Chapter 3.2 in \ci{well:96}.  In particular, it ensures that weak convergence of the empirical process $M_n$ to the population process $M$ in turn implies convergence of their minimizers.  Furthermore, existence follows immediately if $\Om$ is compact.  The conditions on the covering number in (P1) and curvature in (P2) arise from empirical process theory and control the behavior of $M_n - M$ near the minimum in order to obtain rates of convergence.

We also consider uniform convergence results for predictor values $x$, requiring stronger versions of the above assumptions.  Let $B > 0$.
\begin{itemize}
  \item[(U0)] Almost surely, for all $\lVert x \rVert_E \leq B$, the objects $\mp(x)$ and $\hmp(x)$ exist and are unique.  Additionally, for any $\eps > 0$,
      \[
      \inf_{\lVert x \rVert_E \leq B} \inf_{d(\om, \mp(x)) > \eps} M(\om, x) - M(\mp(x), x) > 0
      \]
      and there exists $\zeta = \zeta(\eps) > 0$ such that
      \[
      P\left(\inf_{\lVert x \rVert_E \leq B} \inf_{d(\om, \hmp(x)) > \eps} M_n(\om, x) - M_n(\hmp(x), x) \geq \zeta \right) \ra 1.
      \]
  \item[(U1)] 
  With $B_\delta(m_\oplus(x))$ and $N(\epsilon, B_\delta(m_\oplus(x)), d)$ as in (P1), there exists $\gamma \geq 0$, possibly depending on $B$, such that
  $$
  \int_0^1 \sup_{\lVert x \rVert_E \leq B}\sqrt{1 + \log N(\delta\epsilon, B_\delta(m_\oplus(x)), d)} \;d\epsilon = O(\delta^{-\gamma}) \quad \textrm{as} \quad \delta \rightarrow 0.
  $$
  \item[(U2)] There exist $\tau>0$, $D>0$, and $\alpha > 1$, possibly depending on $B$, such that
  \[
  \inf_{\lVert x \rVert_E \leq B} \inf_{d(\om, \mp(x)) < \tau} \left\{ M(\om, x) - M(\mp(x), x) - Dd(\om, \mp(x))^\alpha\right\} \geq 0.
  \]
\end{itemize}

The following examples of classes of random objects correspond to the applications and simulations that will be discussed in Sections~\ref{sec: density}, \ref{sec: correlation} and \ref{sec: sphere}.

\begin{Example}
\label{exm: wass}
  Take $\Om$ to be the set of univariate probability distributions $G$ supported within a common compact interval $\mathcal{I} \subset \R$. Equip $\Om$ with the Wasserstein metric $d_W$, defined for two such distributions $G_1$ and $G_2$ as
\[
d_W(G_1, G_2)^2 = \int_0^1 (G_1\inv(t) - G_2\inv(t))^2\;dt,
\]
where $G_1\inv$ and $G_2\inv$ are the quantile functions corresponding to $G_1$ and $G_2$, respectively.
\end{Example}

\begin{Example}
\label{exm: correlation}
  Take $\Om$ as the set of correlation matrices of a fixed dimension $r$, i.e. symmetric, positive semidefinite $r\times r$ matrices with unit diagonal, and equip $\Om$ with the Frobenius metric, $d_F$.
\end{Example}

\begin{Example}
  \label{exm: manifold}
  Let $\Om$ be a (bounded) Riemannian manifold of dimension $r$ and let $d$ be the geodesic distance implied by the Riemannian metric.
\end{Example}

Propositions~\ref{prop: wass} and \ref{prop: correlation} in the Appendix 
demonstrate that all of the above assumptions are satisfied for the random objects in Examples~\ref{exm: wass} and \ref{exm: correlation}, with $\beta = \alpha = 2$ in (P2) and (U2).  Example~\ref{exm: wass} differs from Example~\ref{exm: correlation}
in that it is infinite-dimensional, which leads to the value $\gamma = 1/2$ in (P1) and (U1) for Example~\ref{exm: wass}, while $\gamma = 0$ in Example~\ref{exm: correlation}.  Using the entropy-based results of Theorem~\ref{thm: rate} and Corollary~\ref{cor:cor2}
below, the positive value of $\gamma$ for the Wasserstein space implies a rate of convergence that is slower than the parametric rate $n^{-1/2}$ that is available in Example~\ref{exm: correlation}.  This entropy-based rate is,
however, suboptimal: in Section~\ref{sec: hilbert}, a more efficient argument is presented, applicable to the Wasserstein space or, more generally, to any space that is isometric to a closed, convex subset of a separable Hilbert space,
that entirely bypasses the entropy bound and recovers the parametric rate $n^{-1/2}$ for the global Fr\'echet regression estimator and the classical nonparametric rate $n^{-2/5}$ for the local Fr\'echet regression estimator.

We note that Example~\ref{exm: wass} refers to objects in the Wasserstein space, a complex smooth manifold that is characterized by the Wasserstein geodesics \cp{taka:11} and thus provides an example of random objects on a manifold, for which we obtain consistent estimation of global and local Fr\'echet regression as demonstrated below. Example~\ref{exm: correlation} refers to  a convex space and, at first glance, it seems straightforward to implement local regression using kernel weights on such spaces. This is however not the case; a major difficulty is that global and local regression assign negative weights near the boundaries, where the boundary is a very substantial part of the domain especially in the global regression case. 

For Example~\ref{exm: manifold}, Proposition~\ref{prop: manifold} shows that (P1) and (U1) hold automatically and, if (P0) (respectively (U0)) holds, then (P2) (respectively (U2)) is equivalent to the Hessian on the tangent space at $\mp(x)$ being positive definite at $0$, and in this case we may take $\alpha = \beta = 2$. Thus, for manifolds, local curvatures do not influence the convergence rates below.  Uniqueness of Fr\'echet means for manifolds is challenging in general, but can be guaranteed under certain circumstances, for example restricting the support of the underlying distribution $F_Y$ \cp{afsa:11}.  Alternatively, one can consider Fr\'echet mean sets \cp{ziez:77}; see also the last paragraph in Section~\ref{sec: dis}.

The following two results demonstrate the consistency of our proposed estimators and also provide rates of convergence.  All proofs can be found in the Appendix.

\begin{Theorem}
\label{lma: con}
  Suppose (P0) holds and $\Om$ is bounded.  Then for any fixed $x \in \R$, $d(\hmp(x), \mp(x)) = o_p(1)$.  For $B > 0$, if (U0) holds then \newline $\sup_{\lVert x \rVert_E \leq B} d(\hmp(x), \mp(x)) = o_p(1)$.
\end{Theorem}

\begin{Theorem}
\label{thm: rate}
  Suppose that, for a fixed $x \in \R^p$, (P0)--(P2) hold, with $\gamma$ as in (P1).  Then
  $$
  d(\hmp(x), \mp(x)) = O_p\left(n^{-\frac{1}{2(\beta + \gamma - 1)}}\right).
  $$
  Furthermore, for a given $B > 0$, if (U0)--(U2) hold, with $\gamma$ as in (U1),
  $$
  \sup_{\lVert x \rVert_E \leq B} d(\hmp(x), \mp(x)) = O_p\left(n^{-\frac{1}{2(\alpha' + \gamma - 1)}}\right)
  $$
  for any $\alpha' > \alpha.$
\end{Theorem}
In general, the rate of convergence is determined by the local geometry near the minimum as quantified in (P2) and (U2).  The proof of the pointwise result follows along the lines of Theorem 3.2.5 in \ci{well:96} which deals with $M$-estimators, where some additional considerations are needed to deal with the necessary estimation of the mean and covariance of $X$.  The uniform result is more difficult, as an uncountable number of $M$-estimators are considered simultaneously and no parametric form of the regression function is available.  When  $\Om$ has a smooth structure, e.g.,  the Wasserstein space in Example~\ref{exm: wass} or a smooth Riemannian manifold, one can conceivably also obtain a limiting distribution.  We demonstrate this for the case where $\Om$ is a Hilbert space in Section \ref{sec: hilbert}.

\section{Local Fr\'echet Regression}
\label{sec: local}

As the success of nonparametric regression methods over the last decades has shown, there is often the need for local rather than global fitting of regression functions. Local regression is more flexible but on the other hand requires choosing a tuning parameter that balances bias and variance. As far as we know, to date, local estimation of \eqref{eq: frechet_regression} for responses in general metric spaces has been exclusively done with the Nadaraya-Watson estimator \cp{davi:07,hein:09,stei:09,stei:10}
\begin{equation}
\label{eq: NWS}
\hmp^{\textrm{NW}}(x) = \argmin_{\om \in \Om} \frac{1}{n}\sum_{i = 1}^n K_h(X_i - x)d^2(Y_i, \om),
\end{equation}
where $K$ is a smoothing kernel that corresponds to a probability density and $h$ is a bandwidth, with $K_h(\cdot) = h\inv K(\cdot/h).$  In this section, the proposed Fr\'echet regression analysis is extended from the global setting, as described in the previous sections, to a local version.  The idea is to adopt the concepts of local linear regression, which has been established for real-valued responses, and then to extend them to the case where responses are random objects, in analogy to the developments in Section~\ref{ss: frechet} for global Fr\'echet regression.  Thus, we develop a novel local version of smoothing in general metric spaces which goes beyond the Nadaraya-Watson smoother \eqref{eq: NWS}.  As is the case for Euclidean data, this local Fr\'echet regression proves to be superior to Nadaraya-Watson smoothing, especially near the boundaries, as demonstrated in the experiments in Sections~\ref{ss: mortality} and \ref{sec: sphere}. Moreover, our analysis of these estimators separates bias and stochastic variation of the corresponding estimators. 

For ease of representation, we consider here the case of a scalar predictor $X \in \R^p$, where  $p = 1$; the local method can also be developed for any $p$ with  $p>1$. The target is again (\ref{eq: frechet_regression}), where we make no structural assumptions on $\mp$.  Consider the preliminary case $\Om = \R$, and again write $m = \mp$.  In this case, the local linear estimate \cp{fan:96} of $m(x)$ is $\hl(x) = \hbn$, where
\[
(\hbn, \hbo) = \argmin_{\bn, \bo}\frac{1}{n}\sum_{i = 1}^n K_h(X_i - x)(Y_i - \bn - \bo(X_i - x))^2.
\]
In this sense, the estimates $\hbn$ and $\hbo$ can be viewed as $M$-estimators of
\begin{equation}
\label{eq: local_target}
(\sbn, \sbo) = \argmin_{\bn, \bo}\int K_h(z - x)\left[\int y {\rm d}F_{Y|X}(z, y) - (\bn + \bo(z - x))\right]^2{\rm d}F_X(z) .
\end{equation}

Defining $\mu_j = E\left[K_h(X - x)(X-x)^j\right]$, $r_j = E\left[K_h(X - x)(X-x)^jY\right]$ and $\sigma_0^2 = \mu_0\mu_2 - \mu_1^2$, the solutions to (\ref{eq: local_target}) are
\[
\sbn = \sigma_0^{-2}(\mu_2r_0 - \mu_1r_1),\quad \sbo = \sigma_0^{-2}(\mu_0r_1 - \mu_1r_0).
\]
This means that $\hl(x) = \hat{\beta}_0$ can be viewed as an estimator of the intermediate target
\begin{align}
\tl(x) = \sbn &= \frac{\mu_2r_0 - \mu_1r_1}{\sigma_0^2} = \frac{1}{\sigma_0^2}\int yK_h(z - x)\left[\mu_2 - \mu_1(z-x)\right]{\rm d}F(z,y) \label{eq: local_sol2} \\
&= E[s(X, x, h)Y] \nonumber
\end{align}
for the weight function
\[
s(z, x, h) = \frac{1}{\sigma_0^2}\left\{K_h(z - x)\left[\mu_2 - \mu_1(z - x)\right]\right\}.
\]
Observing that $\int s(z, x, h) \dFz \equiv 1,$ it follows that $\tl(x)$ in (\ref{eq: local_sol2}) corresponds to a localized Fr\'echet mean,
\begin{equation}
\tl(x)=\argmin_{y \in \R}E \left[s(X, x, h)(Y - y)^2\right]. \label{eq: local_smooth}
\end{equation}
The minimizer $\tl(x)$  in (\ref{eq: local_smooth}) can be viewed as a smoothed version of the true regression function, with the bias $m(x) - \tl(x) = o(1)$ as $h \ra 0$.  Under mild assumptions on the kernel and distribution $F$, this bias is $O(h^2)$, which follows from a Taylor expansion argument.

Now we are in a position to define the local regression concept for random objects $Y \in \Om$, in analogy to the global Fr\'echet regression.   Specifically, (\ref{eq: local_smooth}) can be generalized by defining $\tL(\om) = E\left[s(X, x, h)d^2(Y, \om)\right]$, where the dependency on $n$ is through the bandwidth sequence $h=h_n$, and then setting
\[
\tlp(x) = \argmin_{\om \in \Om} \tL(\om).
\]
In contrast to Euclidean spaces or Riemannian manifolds \cp{yuan:12}, no version of a Taylor expansion argument is available on general metric spaces $\Om$.  So one can ask why this weighted Fr\'echet mean provides a good approximation to the conditional mean in (\ref{eq: frechet_regression}). It turns out that this  is due to the fact (shown in the proof of Theorem~\ref{thm: local_bias} below) that
\[
\left[\int s(z, x, h){\rm d}F_{X|Y}(z,y)\right]\dy = {\rm d}F_{Y|X}(x, y) + O(h^2),
\]
so that minimizing $\tL$ is approximately the same as minimizing the conditional Fr\'echet function $M_\oplus(\cdot, x)$.

The target $\tlp(x)$ can be estimated by using preliminary estimates $\hmu_j = n\inv\sn K_h(X_i - x)(X_i-x)^j$, $\hat{\sigma}_0^2 = \hmu_0\hmu_2 - \hmu_1^2$, and the empirical weights
\[
s_{in}(x, h) = \frac{1}{\hat{\sigma}_0^2}K_h(X_i - x)\left[\hat{\mu}_2 - \hat{\mu}_1(X_i-x)\right].
\]
Then setting $\hL(\om) = n\inv\sn s_{in}(x, h)d^2(Y_i, \om)$, the local Fr\'echet regression estimate is
\begin{equation}
\label{eq: local_frechet_est}
\hlp(x) = \argmin_{\om \in \Om}\hL(\om).
\end{equation}

While this local estimation technique is developed here for general metric space data, it is of interest to compare it to other local estimators that have been previously considered for spaces with additional structure, specifically the intrinsic local polynomial (ILPR) estimator for manifold data proposed in \ci{yuan:12}, where  covariance matrices as objects are  regressed against scalar predictors.  Whereas the ILPR estimator requires various technical steps involving exponential, logarithmic and parallel transport maps on the manifold, one advantage of the methodology proposed here is its simplicity, only requiring distances between data objects.  In terms of computation on manifolds, the current method also enjoys the distinct advantage of requiring optimization only for a single object, unlike the ILPR for which one has to fit both intercept and derivative terms. It is of course also much more general, providing consistent estimators in unstructured metric spaces.  Furthermore, the function to be minimized is merely a weighted least squares problem, potentially with negative weights.  Thus, any metric space for which a Nadaraya-Watson smoother \cp{hein:09} is computationally feasible, or any manifold for which the ILPR can be computed, is also feasible for both local and global Fr\'echet regression.  In the manifold case, expressions for the Riemannian gradient and Hessian are available for a variety of complex manifolds \cp{ferr:13:2}, which can be used for Newton-type algorithms, possibly in conjunction with stochastic optimization techniques, such as the annealing algorithm of \ci{yuan:12}.

For  a concrete comparison of local Fr\'echet regression with the ILPR, take $\Om$ to be the space of covariance matrices with $d$ being the Log-Euclidean metric, that is, $d(\om_1,\om_2) = d_F(\rm{Log} \,\om_1, \rm{Log}\,  \om_2)$, where $d_F$ is the Frobenius metric and $\rm{Log}$ is the inverse of the matrix exponential \rm{Exp} \cp{arsi:07}.   In this case, both the ILPR and local Fr\'echet regression estimates can be computed analytically.  
For a sample $(X_i, Y_i)$, with $Y_i$ a positive definite covariance matrix, both methods yield the estimate
\[
\hmp(x) = \rm{Exp}\left(\frac{\hmu_2\hat{r}_0 - \hmu_1\hat{r}_1}{\hat{\sigma}_0^2}\right),
\]
where $\hat{r}_j = n\inv\sum_{i = 1}^n K_h(X_i-x)(X_i-x)^j{\rm{Log}}(Y_i)$. That these two methods coincide is not altogether surprising due to the metric being the Euclidean metric on transformed matrices.  However, it shows that in this situation local Fr\'echet regression gives a sensible and intuitive estimate which coincides with the previously established manifold-based estimator.

Returning  to theory, in order to obtain the rate of convergence for the quantity $d(\mp(x), \hlp(x))$, we need to quantify the convergence of the bias term $d(\mp(x), \tlp(x))$ and the stochastic term $d(\tlp(x), \hlp(x))$.  This requires the assumptions below.  Recall that \mbox{$M_\oplus(\om, x) = E(d^2(Y, \om)|X = x)$}.  For simplicity, we assume that the marginal density $f$ of $X$, within the joint distribution $F$, has unbounded support, and consider points $x \in \R$ for which $f(x) > 0$. We need the following assumptions.
\begin{itemize}
  \item[(K0)] The kernel $K$ is a probability density function, symmetric around zero. Furthermore, defining $K_{kj} = \int_\R K^k(u)u^j \; du$, $|K_{14}|$ and $|K_{26}|$ are both finite.
  \item[(L0)]  The object $\mp(x)$ exists and is unique.  For all $n$, $\tlp(x)$ and $\hlp(x)$ exist and are unique, the latter almost surely.  Additionally, for any $\eps > 0$,
  \begin{align*}
  \inf_{d(\om, \mp(x)) > \eps} \left\{M_\oplus(\om, x) - M_\oplus(\mp(x), x)\right\}>0, \\
  \liminf_n \inf_{d(\om, \tlp(x)) > \eps}\left\{\tL(\om) - \tL(\tlp(x))\right\}>0.
  \end{align*}
  \item[(L1)]  The marginal density $f$ of $X$, as well as the conditional densities $g_y$ of \mbox{$X|Y = y$}, exist and are twice continuously differentiable, the latter for all $y \in \Om$, and $\sup_{x,y} |g_y''(x)| < \infty$.  Additionally, for any open $U\subset \Om$, $\int_U {\rm d}F_{Y|X}(x, y)$ is continuous as a function of $x$.
  \item[(L2)]  There exists $\eta_1 > 0$, $C_1 > 0$ and $\beta_1 > 1$ such that
  \[
  M_\oplus(\om, x) - M_\oplus(\mp(x), x) \geq C_1d(\om, \mp(x))^{\beta_1},
  \]
  provided $d(\om, \mp(x)) < \eta_1$.
  \item[(L3)]  There exists $\eta_2 > 0$, $C_2 > 0$ and $\beta_2 > 1$ such that
  \[
  \liminf_n \left[\tL(\om) - \tL(\tlp(x))\right] \geq C_2d(\om, \tlp(x))^{\beta_2},
  \]
  provided $d(\om, \tlp(x)) < \eta_2$.
\end{itemize}

Assumptions (K0) and (L1) are common in local regression estimation and imply that the smoothed marginal distribution $$\left(\int s(z, x, h) {\rm d}F_{X|Y}(z|y)\right)\dy$$ converges to ${\rm d}F_{Y|X}(x, y)$ as $h \ra 0$, while (L2) and (L3) provide the rate for the bias and stochastic terms, respectively.  While (L1) is a distributional assumption, (L2) and (L3)  can be shown to hold for Examples~\ref{exm: wass}--\ref{exm: manifold} in Section~\ref{sec: theory}, using arguments similar to those in Propositions~\ref{prop: wass}--\ref{prop: manifold} in the Appendix. In these cases, it is easy to verify that $C_j=1$, $\beta_j = 2$ and $\eta_j$ arbitrary, $j = 1,2$, are admissible in (L2) and (L3).  We now state our main results for local Fr\'echet regression, where the first result is  for the bias, the second for  the stochastic deviation and the corollary  combines these results to obtain an overall rate of convergence. 

\begin{Theorem}
  \label{thm: local_bias}
  If (K0), (L0), (L1), (L2) and (P1) hold, then \[d(\mp(x), \tlp(x)) = O(h^{2/(\beta_1 - 1)})\] as $h = h_n \ra 0$.
\end{Theorem}
\begin{Theorem}
  \label{thm: local_frechet}
  If (K0), (L0), (L1), (L2), (L3) and (P1) hold, with $\gamma$ as in (P1), and if $h \ra 0$ and $nh \ra \infty$ with $nh^{(\beta_1 + 4\beta_2 + 4\gamma - 5)/(\beta_1 - 1)} = O(1)$, then
  \[
  d(\tlp(x), \hlp(x)) = O_p\left[ (nh)^{-\frac{1}{2(\beta_2 + \gamma - 1)}}\right].
  \] \end{Theorem}

  \begin{Corollary} \label{cor:cor2} Under the assumptions of Theorem~\ref{thm: local_bias} and Theorem~\ref{thm: local_frechet}, among bandwidth sequences
  $h = n^{-\kappa}$ satisfying the hypotheses, the optimal sequence is
   obtained for  \mbox{$\kappa^* = (\beta_1 - 1)/(\beta_1 + 4\beta_2 + 4\gamma - 5)$} and yields the rate  $$d(\mp(x), \hlp(x)) = O_p\left(n^{-2/(\beta_1+4\beta_2+4\gamma-5)}\right).$$
\end{Corollary}
We note that for $\beta_1=\beta_2=2$, and $nh^{5+4\gamma} = O(1)$, one obtains the result $$d(\mp(x), \hlp(x)) = O_p\left[ (nh)^{-\frac{1}{2(1+\gamma)}} + h^2\right]$$ with optimal bandwidth exponent $\kappa^* = 1/(5+4\gamma)$ and 
corresponding rate $d(\mp(x),\hlp(x)) = O_p(n^{-2/(5+4\gamma)})$.  When $\gamma = 0$, as is the case for real-valued responses, this recovers the result that is familiar for local regression with real valued responses, and with $\kappa^* =1/5$ leads to the rate  $d(\mp(x), \hlp(x)) = O_p(n^{-2/5}).$  For the Wasserstein space with $\gamma = 1/2$ (see Proposition~\ref{prop: wass} below), this generic entropy-based argument instead yields the slower rate $n^{-2/7}$; Section~\ref{sec: hilbert} shows that this rate is suboptimal.  While the above results are pointwise, we remark that a uniform rate over $x$ in a bounded interval can be obtained by suitably strengthening assumptions (L0), (L2) and (L3), similar to the global case.

\section{Limiting Distributions in Hilbert Spaces and Rates for the Wasserstein Space}
\label{sec: hilbert}

A case of particular interest arises when the random objects are functions that are assumed to be (almost surely) square-integrable, e.g., $\Om = L^2[0,1]$ \cp{fara:97}.  
Going beyond functional data as responses, we more generally assume that $\Om$ is a separable Hilbert space with inner product $\langle \cdot, \cdot \rangle$ and corresponding norm $\lVert \cdot \rVert_\Om$.  As before, let $F$ be a distribution on $\mathbb{R}\times \Om$ with $(X, Y) \sim F$.  
As this setting enables linear operations, the minimizing objects in (\ref{eq: mp}) and (\ref{eq: mp_est}) can be given explicitly under mild assumptions on the moments of $F.$  Unsurprisingly, for the case of functional data, the minimizer 
of (\ref{eq: mp_est}) corresponds to the estimator given in \ci{fara:97}.  Global and local Fr\'echet regression in Hilbert spaces are treated in Sections~\ref{ss: hilbert_global} and \ref{ss: hilbert_local}, respectively.  In Section~\ref{ss: wass_hilbert}, these results are applied to the Wasserstein space of Example~\ref{exm: wass} through the quantile-function representation, yielding rates of convergence for global and local Fr\'echet regression that improve on those obtained from the general results of Sections~\ref{sec: theory} and \ref{sec: local}.  
  
\subsection{Global Fr\'echet Regression in Hilbert Spaces}
\label{ss: hilbert_global}
We will use the following notation.  For $q > 1$, let $\Om^q$ be the $q$-fold Cartesian product of $\Om$, with inner product $\ip{\alpha}{\alpha'}_q = \sum_{l = 1}^q\ip{\alpha_l}{\alpha'_l}$ for $\alpha,\alpha'\in\Om^q$, so that $\Om^q$ is also a Hilbert space.  For  a $p\times p$ matrix $A$, $x \in \R^p$, $\om \in \Om$ and $\alpha\in \Om^p$, we define $A\alpha \in \Om^p$ with elements $(A\alpha)_l = \sum_{m = 1}^p A_{lm}\alpha_m$, $\alpha^Tx = \sum_{l = 1}^p x_l\alpha_l \in \Om$ and $x\om \in \Om^p$ with elements $(x\om)_l = x_l\om$.

\begin{Theorem}
\label{thm: hilbert1}
A. Let $(X, Y) \sim F$ and suppose that $E \lVert Y \rVert_\Om^2 < \infty$.  Then there exist unique elements $\gamma_0 \in \Om$ and $\gamma_1 \in \Om^p$ which satisfy, for all $\om \in \Om$ and $\alpha \in \Om^p$,
$$
E\ip{Y}{\om} = \ip{\gamma_0}{\om} \quad \textrm{and}\quad  E\ip{(X - \mu)Y}{\alpha}_p = \ip{\gamma_1}{\alpha}_p.
$$
With $\Sigma = \Var(X)$ and defining  $\beta_1 := \Sigma\inv\gamma_1$ and $\beta_0 = \gamma_0$, the  solution to (\ref{eq: mp}) is
\begin{equation}
\label{eq: mp_hilsol}
\mp(x) = \beta_0 + \beta_1^T(x - \mu).\end{equation}

B. Define estimators  $\hat{\gamma}_0 = \bar{Y} = n\inv\sn Y_i$, \mbox{$\hat{\gamma}_1 = n\inv\sn(X_i - \xbar)Y_i$}, $\hbeta_1 = \hSi\inv\hat{\gamma}_1$ and $\hbeta_0 = \hat{\gamma}_0$. 
The solution of (\ref{eq: mp_est}) is given by
\begin{equation}
\label{eq: mp_est_hilsol}
\hmp(x) = \hat{\beta}_0 + \hbeta_1^T(x-\xbar).
\end{equation}
\end{Theorem}

Results \eqref{eq: mp_hilsol} and \eqref{eq: mp_est_hilsol} demonstrate that explicit solutions of the minimization problems that define the global Fr\'echet regression are available for the case of responses that are 
random objects in Hilbert space. Moreover, in this situation one can also obtain limiting distributions, as follows.

\begin{Theorem}
\label{thm: hilbert2}
Set $\beta = (\beta_0, \beta_1^T)^T$ and $\hat{\beta} = (\hat{\beta}_0, \hat{\beta}_1^T)^T$.  Under the assumptions of Theorem~\ref{thm: hilbert1}, and if, in addition, $E(\lVert X\rVert_E^2\lVert Y\rVert_\Om^2) < \infty$,
\[
\sqrt{n}(\hat{\beta} - \beta) \rightsquigarrow \mathcal{G},
\]
where `$\rightsquigarrow$' denotes weak convergence and $\mathcal{G}$ is a zero mean Gaussian process on $\Om^{p+1}$.  The covariance structure of $\mathcal{G}$ is defined by projection covariances $\Cov(\ip{\mathcal{G}}{\alpha}_{p+1}) = l_\alpha^TC_\alpha l_\alpha$, where $\alpha \in \Om^{p+1}$, $C_\alpha$ is the covariance matrix of the vector defined in (\ref{eq: Z_vector}) 
in the Appendix and $l_\alpha$ can be constructed using the expressions in (\ref{eq: l_alpha}) 
in the Appendix.
\end{Theorem}

We next consider weak convergence of the process
\[
\mathcal{M}_n(x) = \sqrt{n}(\hmp(x) - \mp(x))
\]
as $x$ varies in $\R^p$.  For any $U\subset \R^p$, define the function space $$l_\Om^\infty(U) = \{g:U \ra \Om: \sup_{x \in U}\lVert g(x)\rVert_\Om < \infty\}$$ with norm $\lVert g \rVert_U = \sup_{x \in U}\lVert g \rVert_\Om$.  

\begin{Corollary}
  \label{cor: hilbert}
 Let $B > 0$ be arbitrary, and define $V_B = \{x \in \R^p: \lVert x \rVert_E \leq B\}$.  Under the assumptions of Theorem~\ref{thm: hilbert2},
    \[
    \sup_{x\in V_B} \lVert \hmp(x) - \mp(x)\rVert_\Om = O_p(n^{-1/2}).
    \]
Additionally, there is a zero-mean Gaussian process $\mathcal{M}$ on $V_B$ such that $$\mathcal{M}_n \rightsquigarrow \mathcal{M} \quad \text{in} \quad  l_\Om^\infty(V_B),$$ where $\mathcal{M}_n$ is restricted to $V_B$.
\end{Corollary}
These results show that one can take advantage of the additional structure that is available in the case of Hilbertian objects to obtain limit distributions of the estimates.  Limit distributions are not available for general object spaces due to the lack of a linear structure.  Generally, even for the simpler case of
Fr\'echet means, limit results cannot be directly obtained, except in special cases. For example, for random objects that fall on manifolds satisfying certain regularity conditions, local linear approximations sometimes make it possible to derive limit theorems 
\cp{bard:13}.

\subsection{Local Fr\'echet Regression in Hilbert Spaces}
\label{ss: hilbert_local}

In the case of local \newline 
Fréchet regression, the objective functions $\tL$ and $\hL$ are weighted least-squares criteria, similar to the global case, with weights $s(X,x,h)$ and $s_{in}(x,h)$, respectively, in 
place of $s(X,x)$ and $s_{in}(x)$.  As in Section~\ref{sec: local}, we consider a scalar predictor $x \in \R$ that is an interior point.  Throughout this subsection, as in Section~\ref{sec: local}, $\mp(x)$ denotes the Fr\'echet regression function in (\ref{eq: frechet_regression}), which in a Hilbert space is the conditional mean $E(Y \mid X = x)$, rather than the global Fr\'echet regression function in (\ref{eq: mp}).

\begin{Theorem}
\label{thm: hilbert_local}
A. Suppose (K0) and (L1) hold, $K$ is bounded with compact support, $f(x) > 0$, and $E \lVert Y \rVert_\Om^2 < \infty$. Then for sufficiently small bandwidth $h > 0$, there is a unique element $\tlp(x) \in \Omega$ 
satisfying $\ip{\tlp(x)}{\omega} = E[s(X,x,h)\ip{Y}{\omega}]$ for all $\om \in \Om$, and $\tlp$ is the unique minimizer of $\tL$.

B. Under the assumptions of A, suppose further that $h \rightarrow 0$ and $nh\rightarrow \infty$ as $n \rightarrow \infty$. Then with probability approaching 1, $\hlp(x) = n^{-1}\sum_{i=1}^n s_{in}(x,h)Y_i$ is the unique minimizer of $\hL$.
\end{Theorem}

\begin{Theorem}
\label{thm: hilbert_local2}
Suppose that the assumptions of Theorem~\ref{thm: hilbert_local} hold, and that $h \ra 0$ and $nh \ra \infty$.  Suppose further that, for
some $\eps > 0$ and some neighborhood $U$ of $x$,
\begin{enumerate}[(i)]
\item $\sup_{z \in U} E(\lVert Y\rVert_\Om^{2+\eps} \mid X = z) < \infty$;
\item $z \mapsto \mp(z) = E(Y \mid X = z)$ is continuous at $x$;
\item $z \mapsto \Gamma(z) = E[(Y - \mp(z))\otimes(Y - \mp(z)) \mid X = z]$ is continuous at
$x$ in trace norm, where $(u\otimes v)(w) = \ip{v}{w}u$.
\end{enumerate}
Then as $n \ra \infty$,
\[
\sqrt{nh}\left(\hlp(x) - \tlp(x)\right) \rightsquigarrow \mathcal{G}_x,
\]
where $\mathcal{G}_x$ is a zero-mean Gaussian element of $\Om$ with covariance operator
$R(K)\Gamma(x)/f(x)$, $R(K) = \int_\R K^2(u)\,du$; that is, for all $\om \in \Om$,
\[
\Var(\ip{\mathcal{G}_x}{\om}) = \frac{R(K)}{f(x)}\Var(\ip{Y}{\om}\mid X = x).
\]
\end{Theorem}

\noindent\textbf{Remark.} Because Theorem~\ref{thm: hilbert_local2} centers $\hlp(x)$ at the
smoothed target $\tlp(x)$, no bias term arises and no condition beyond $nh \ra \infty$ is
imposed on the bandwidth; in particular, $\lVert \hlp(x) - \tlp(x)\rVert_\Om = O_p((nh)^{-1/2})$.
If, in addition, $\mp$ is twice differentiable as a map into $\Om$ in a
neighborhood of $x$, then writing $\tlp(x) - \mp(x) = E[s(X,x,h)\{\mp(X) - \mp(x)\}]$, a
second-order Taylor expansion of $\mp$ about $x$, together with the identities
$E[s(X,x,h)(X-x)] = 0$ and $E[s(X,x,h)(X-x)^2] = K_{12}h^2 + o(h^2)$
(see Lemma~\ref{lma: muj}), gives
\[
\tlp(x) - \mp(x) = \tfrac{1}{2}K_{12}\,\mp''(x)\,h^2 + o(h^2).
\]
Consequently, under the undersmoothing condition $nh^5 \ra 0$, the smoothed target may be
replaced by the conditional mean, so that
$$
\sqrt{nh}\left(\hlp(x) - \mp(x)\right) \rightsquigarrow \mathcal{G}_x.
$$
This comes at the
price of a slower rate of convergence, since $nh^5 \ra 0$ excludes the choice
$h \sim n^{-1/5}$.  Instead, balancing the stochastic and bias terms 
with $h \sim n^{-1/5}$ yields $\lVert \hlp(x) - \mp(x)\rVert_\Om = O_p(n^{-2/5})$.

\subsection{Application to the Wasserstein Space}
\label{ss: wass_hilbert}

We now show that, for the Wasserstein space of Example~\ref{exm: wass}, both the global and local Fr\'echet regression estimators attain, respectively, standard parametric and nonparametric rates of convergence, 
improving on the conclusion drawn by combining entropy-based results of Theorems~\ref{thm: rate} and \ref{thm: local_frechet} with the properties established in Proposition~\ref{prop: wass} for the Wasserstein space.

Consider first the global model.  Let $Q(\om)$ denote the quantile function of $\om \in \Om$ and $Q\inv$ the inverse map, mapping quantile functions to their associated distribution.  Set $\hat{g}_x = n\inv\sn s_{in}(x)Q(Y_i) \in L^2[0,1]$, with $s_{in}(x)$ as in (\ref{eq: emp_weights}), and let $d_{L^2}$ be the standard $L^2$ metric.  As shown in the proof 
of Proposition~\ref{prop: wass} in the Appendix, the global Fr\'echet regression estimator has the explicit form
\begin{equation}
  \label{eq: wass_hilbert}
\hmp(x) = \argmin_{\om \in \Om} d_{L^2}^2(\hat{g}_x, Q(\om)) = Q\inv\left(\argmin_{q \in Q(\Om)} d_{L^2}^2(\hat{g}_x, q)\right).
\end{equation}
For a nonempty, closed, convex set $\mathcal{C} \subset L^2[0,1]$, denote the
metric projection of $h \in L^2[0,1]$ onto $\mathcal{C}$ by $\Pi_{\mathcal{C}}(h) := \argmin_{q \in \mathcal{C}} d_{L^2}(h, q)$, which exists and is
unique. Since $Q(\Om)$ is closed and convex, \eqref{eq: wass_hilbert} states that
$Q(\hmp(x)) = \Pi_{Q(\Om)}(\hat{g}_x)$; by an analogous argument in the population setting,
$Q(\mp(x)) = \Pi_{Q(\Om)}(g_x)$ with $g_x = E[s(X,x)Q(Y)]$.  Instead of viewing the observed distributions $Y_i$ as elements of $\Om$, 
one can formulate an \emph{extrinsic} Fr\'echet regression model based on viewing the data $\tilde{Y}_i = Q(Y_i)$ as elements of the Hilbert space $L^2[0,1]$.  Since the quantile functions in $Q(\Om)$ are uniformly bounded, the moment
conditions on $\tilde{Y}$ required by Theorems~\ref{thm: hilbert1} and \ref{thm: hilbert2} hold.  The extrinsic model has population and empirical global Fr\'echet regression
targets $\mp^{\rm ext}(x) = g_x$ and $\hmp^{\rm ext}(x) = \hat{g}_x$, respectively,
which can be seen by applying Theorem~\ref{thm: hilbert1} to $\tilde{Y}$. Because the weights
$s(X,x)$ and $s_{in}(x)$ can be negative, the original and extrinsic model targets are related via $Q(\mp(x)) = \Pi_{Q(\Om)}(\mp^{\rm ext}(x))$ and
$Q(\hmp(x)) = \Pi_{Q(\Om)}(\hmp^{\rm ext}(x))$.  By Corollary~\ref{cor: hilbert}, 
applied to the extrinsic model, $\sup_{\lVert x \rVert_E \leq B}d_{L^2}(\hmp^{\rm ext}(x), \mp^{\rm ext}(x)) = O_p(n^{-1/2})$ 
for any $B > 0$.  Since $\Pi_{Q(\Om)}$ is nonexpansive,
\[
\begin{split}
d_W(\hmp(x), \mp(x)) &= d_{L^2}(Q(\hmp(x)), Q(\mp(x))) \\
&= d_{L^2}\big(\Pi_{Q(\Om)}(\hmp^{\rm ext}(x)), \Pi_{Q(\Om)}(\mp^{\rm ext}(x))\big) \\
&\leq d_{L^2}(\hmp^{\rm ext}(x), \mp^{\rm ext}(x)) = O_p(n^{-1/2}),
\end{split}
\]
uniformly for $\lVert x \rVert_E \leq B$. 

Turning to the local estimator, let $\mp(x)$ now denote the conditional Fr\'echet mean
in \eqref{eq: frechet_regression}.  Since $Q(\Om)$ is closed and convex, it follows that 
$Q(\mp(x)) = q_x = E[Q(Y) \mid X = x]$, so no projection is needed at the population level, unlike 
the situation in the global case.  With
$\hat{l}^{\rm ext}(x) = n\inv\sn s_{in}(x,h)\tilde{Y}_i$ being the extrinsic local 
Fr\'echet regression estimator in Theorem~\ref{thm: hilbert_local}B, the same argument 
as for the global model gives $Q(\hlp(x)) = \Pi_{Q(\Om)}(\hat{l}^{\rm ext}(x))$.
If $z \mapsto q_z$ is twice differentiable in a neighborhood of $x$,
and the assumptions of Theorem~\ref{thm: hilbert_local2} hold for $\tilde{Y}$,
then Theorem~\ref{thm: hilbert_local2} and the subsequent Remark, applied to $\tilde{Y}$
with $h \sim n^{-1/5}$, give
\[
\begin{split}
d_W(\hlp(x), \mp(x))
&= d_{L^2}\big(\Pi_{Q(\Om)}(\hat{l}^{\rm ext}(x)), q_x\big) \\
&\leq d_{L^2}\left(\hat{l}^{\rm ext}(x), q_x\right) = O_p(n^{-2/5}).
\end{split}
\]
This local rate matches the classical optimal rate for local linear regression with
real-valued (or Hilbertian) responses, and is markedly faster than the $n^{-2/7}$
rate obtained from the generic entropy-based argument of Corollary~\ref{cor:cor2}
with $\gamma = 1/2$ (see the discussion following Examples~\ref{exm: wass}--\ref{exm: correlation}
above).  As in the global case, the argument uses neither finite dimensionality nor
entropy bounds; it requires only the conditions of Theorem~\ref{thm: hilbert_local2}
and the subsequent Remark.

\section{Fr\'echet Regression for Probability Distributions with the Wasserstein metric}
\label{sec: density}
\subsection{Computational Details}
\label{ss: comp_density}
Here, the space $\Om$ is the set of distribution functions equipped with the Wasserstein metric, as outlined in Example~\ref{exm: wass} in Section~\ref{sec: theory}.  To implement the minimization in \eqref{eq: wass_hilbert} using a sample $(X_i, Y_i)$, $i = 1, \ldots, n$, of covariates and distributions,
let $u_j$, $\,j = 1, \ldots, M,$ be an equispaced grid on $[0,1]$ and let $g_j = \hat{g}_x(u_j)$.  Then compute
\[
q^\ast = \argmin_{q \in \R^M} \lVert g-q\rVert_E^2,
\]
subject to the constraint $q_1 \leq \cdots \leq q_M$.  This optimization problem is a quadratic program and can be solved using a variety of techniques. The solution $q^\ast$ represents a discretized version of the approximation of the quantile function $Q(\hmp(x))$.  Similar arguments hold for the computation of the local Fr\'echet estimator.

\subsection{Simulation}
\label{ss: wass_sim}

To assess the performance of the global Fr\'echet regression estimator in (\ref{eq: mp_est}), it is first necessary to determine a generative model that produces suitably simulated data.  The space of distributions with the Wasserstein metric provides an ideal setting for this.  The responses $Y$ are distributions with quantile functions $Q(Y)$ and the predictors are random variables $X \in \R$.  For notational simplicity, the quantile function corresponding to $Y$ will also be denoted as $Y$.  The regression function is
\[
\mp(x)(\cdot) = E(Y(\cdot)|X=x) =\mu_0 + \beta x + (\sigma_0 + \gamma x)\Phi\inv(\cdot),
\]
where $\Phi$ is the standard normal distribution function, $\mu_0$, $\beta \in \R$ and $\sigma_0$ and $\gamma$ satisfy $\sigma_0 + \gamma x > 0$ for all $x$ in the support of $F_X.$  This corresponds to the response distributions being, on average, a normal distribution with parameters that depend linearly on $x$.

The random response $Y$ is generated conditional on $X$ by adding noise to the quantile functions, which we will demonstrate in two settings.  In the first, the distribution parameters $\mu|X\sim \mathcal{N}(\mu_0 + \beta X, v_1)$ and $\sigma|X\sim \textrm{Gam}((\sigma_0 + \gamma X)^2/v_2, v_2/(\sigma_0 + \gamma X))$ are independently sampled, and the corresponding distribution is $Y = \mu + \sigma\Phi\inv$.  In the second setting, after sampling the distribution parameters as in the previous setting, the resulting distribution is  ``transported" in Wasserstein space following a simplified version of the algorithm outlined in Section 8.1 of \ci{pana:16}. Specifically, random transport maps $T$ (increasing diffeomorphisms of the real line) are generated by sampling uniformly from the collection of transport maps $T_k(x) = x - \sin(kx)/|k|$, for $k \in \{-l, \ldots, l\}\setminus\{0\}$, with $Y = T\circ(\mu + \sigma\Phi\inv)$.  This second setting is significantly more complex, as the observed distributions are no longer Gaussian.

Random samples of pairs $(X_i, Y_i)$, $i = 1, \ldots, n$ were generated by sampling $X_i \sim \mathcal{U}(-1,1)$, setting $\mu_0 = 0$, $\sigma_0 = 3$, $\beta = 3$ and $\gamma = 0.5$, and following the above procedure for the two simulation settings.  In the first setting, the parameter variances were set at $v_1 = 0.25$ and $v_2 = 1$.  In the second, the values were $v_1 = 1$ and $v_2 = 2$, with $l = 2$ used for generating the transport maps.  In each setting, $200$ runs were executed for three sample sizes $n = 50, 100, 200$.  For the $r$-th simulation of a particular sample size, with $\hmp^r(x)$ denoting the fitted distribution function, the quality of the estimation was measured quantitatively by the integrated squared errors
\[
\text{ISE}_r=\int_{-1}^1 d_W^2(\hmp^r(x), \mp(x)) \; {\rm d}x.
\]

In the first simulation setting, we verify that global Fr\'echet regression is performing as expected by comparing to the best-case scenario where one knows the finite-dimensional generating model.  That is, we compute the mean $\mu_i$ and standard deviation of $\sigma_i$ of the distribution $Y_i$ and regress them linearly against $X_i$, while restricting the estimates of $\sigma_0$ and $\gamma$ such that the regression line is positive on $[-1,1]$.  Thus, we can compare this ``oracle" linear regression with global Fr\'echet regression by computing its integrated squared error for each simulation run.  These errors are shown for both methods in boxplots in Figure~\ref{fig: ise_boxplot_low}. It is clear that global Fr\'echet regression performs just as well as the oracle procedure.  Sign-rank tests were performed to test the hypothesis of no difference between the methods for each sample size, with the smallest of the three $p$-values being $0.51$.

\begin{figure}
\centering
\subcaptionbox{\label{fig: ise_boxplot_low}}[2.4in]{\includegraphics{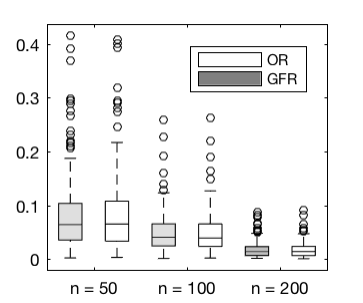}}
\subcaptionbox{\label{fig: ise_boxplot_hi} }[2.4in]{\includegraphics{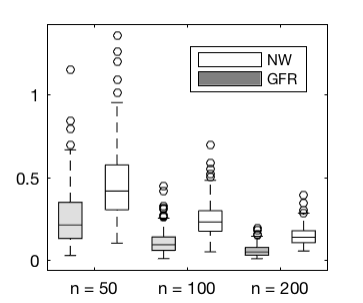}}
\caption{\label{fig: ise_boxplot} Boxplots of integrated squared errors for $200$ simulation runs and three sample sizes $n$.  The left panel compares global Fr\'echet regression (GFR) with the oracle linear regression (OR), while the right shows results for global Fr\'echet regression and the Nadaraya-Watson smoother (NW).}
\end{figure}

In the second simulation setting, the random transportation renders the oracle linear regression technique above inadmissable, since the standard deviation of the transported distribution no longer has a linear relationship with $X$.  However, the global Fr\'echet regression model still holds true.  Figure~\ref{fig: ise_boxplot_hi} shows the decreasing integrated squared errors for increasing sample sizes, demonstrating the validity and utility of global Fr\'echet regression for this complex regression setting.  Furthermore, at the suggestion of a referee, we compared our results with the Nadaraya-Watson estimator in \eqref{eq: NWS}, where the bandwidth was chosen in the interval $[0.2, 0.7]$ so as to minimize the average ISE over all simulations.  This resulted in bandwidth choices 0.5, 0.45, and 0.35 for $n = 50, 100,$ and 200, respectively.  The corresponding ISE values in Figure~\ref{fig: ise_boxplot_hi} demonstrate that this approach is inferior to the global Fr\'echet fits, which is expected if the global model holds, analogous to the situation in the Euclidean case. 

\subsection{Application to Mortality Profiles}
\label{ss: mortality}
Many studies and analyses have been motivated by a desire to understand human longevity.  Of particular interest is the evolution of the distributions of age-at-death over calendar time.  The Human Mortality Database provides such data in the form of yearly lifetables, differentiated by country.  Currently, this database includes yearly mortality and population data for 37 countries that are available at \url{<www.mortality.org>}.
For a given country and calendar year, the probability distribution for mortality can be represented by its density. A first step is to estimate this density from the data in the lifetables for a specified country.  Consider a country for which lifetables are available for the years $t_i$, $i = 1, \ldots, n$.  For integer-valued ages $j$, $j = 0, \ldots, 110$, the lifetable provides the size of the population $m_j$ which is at least $j$ years old, normalized so that $m_0 = 100000$.  These values can be used to construct a histogram for age-at-death, which in turn can be smoothed using a local linear smoother to obtain an estimate of the density.  This smoothing step was performed in \texttt{Matlab} using the \texttt{hades} package, available at \url{<http://www.stat.ucdavis.edu/hades/>}.   Each density was estimated for ages in the interval $[20, 110]$, 
with the value 2 as a common smoothing bandwidth.

\begin{figure}[t]
\centering
\subcaptionbox{\label{fig: chile_dens}}[1.5in]{\includegraphics[width = 1.5in, height = 1.5in]{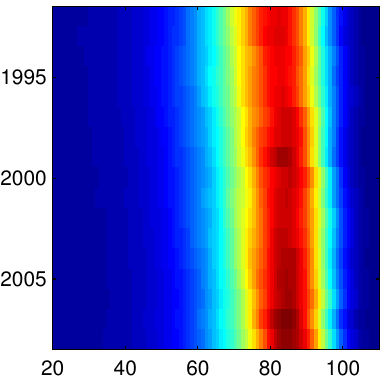}}
\subcaptionbox{\label{fig: chile_lin}}[1.5in]{\includegraphics[width = 1.5in, height = 1.5in]{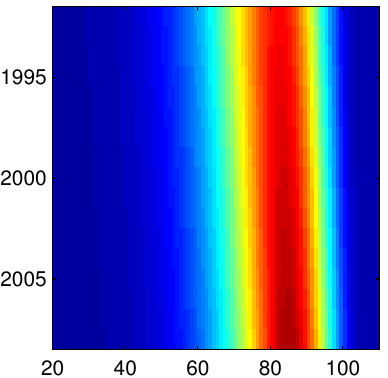}}
\subcaptionbox{\label{fig: chile_quad}}[1.5in]{\includegraphics[width = 1.5in, height = 1.5in]{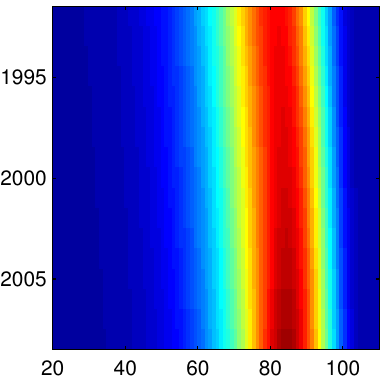}}
\caption{(a) Yearly mortality densities for Chile for the years 1992--2008; (b) Global Fr\'echet regression fits of yearly mortality densities using $X_i = t_i$; (c) Global Fr\'echet regression fits using $\alpha_i = (t_i, t_i^2)^T$. \label{fig: chile}}
\end{figure}

As an initial example, we consider the data for Chile, which has mortality data available for the years 1992--2008.  Using the procedure outlined above, mortality density estimates $Y_i$ were obtained for the years $t_i = 1991 + i$, $i = 1, \ldots, 17$.  These estimates are shown as a heat map in Figure~\ref{fig: chile}, linearly interpolating between years for continuity. The variation from year to year is marked by a steady increase in both the location and height of the peak in mortality.  The global Fr\'echet regression fits using calendar year as predictor  for linear ($X_i = t_i$) and quadratic ($X_i = (t_i, t_i^2)^T$) models are shown in Figures~\ref{fig: chile_lin} and \ref{fig: chile_quad}, respectively.  Similar to the least squares regression plane, these fits provide a smooth visualization of the evolution of mortality and remove the noise that is visible in the raw density data.  There seems to be little gain in fitting a quadratic model, as the global Fr\'echet regression fits with linear and quadratic predictors are very similar.  Leave-one-out prediction errors were 0.088 for the linear fit and 0.0972 for the quadratic fit, indicating that the simpler linear model is indeed preferable.

Next, we consider the data for Luxembourg, with mortality lifetable data ranging from 1960--2009.  The density estimates for these years are shown in Figure~\ref{fig: lux_dens}.  We find a slightly more complicated evolution of mortality for Luxembourg compared to Chile.  For example, the mode of the density does not steadily increase over the years; rather, the mode seems to carve out a curved path.  Figure~\ref{fig: lux_lin} and \ref{fig: lux_quad} show the global Fr\'echet regression fits for the linear and quadratic global Fr\'echet model, with $X_i = t_i$ for the linear and $X_i = (t_i, t_i^2)^T$ for the quadratic model.  The quadratic fit is better at capturing the shape of the peak dynamics observed in the raw sample of densities.  The adjusted Fr\'echet $R^2$ values (for details on these extensions of the coefficient of determination see Section~\ref{ss:inf})  are 0.971 and 0.975 for the linear and quadratic models, respectively. 
Average leave-one-out prediction errors were 0.56 for the linear and 0.27 for the quadratic model.  

\begin{figure}[t]
\centering
\subcaptionbox{\label{fig: lux_dens}}[1.5in]{\includegraphics[width = 1.5in, height = 1.5in]{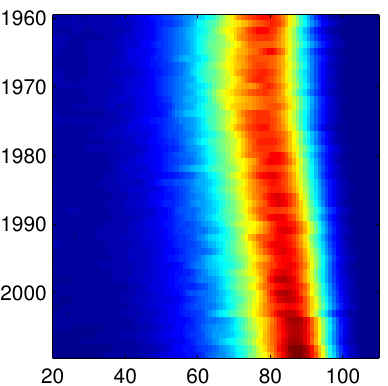}}
\subcaptionbox{\label{fig: lux_lin}}[1.5in]{\includegraphics[width = 1.5in, height = 1.5in]{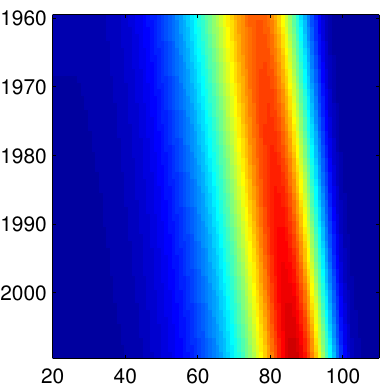}}
\subcaptionbox{\label{fig: lux_quad}}[1.5in]{\includegraphics[width = 1.5in, height = 1.5in]{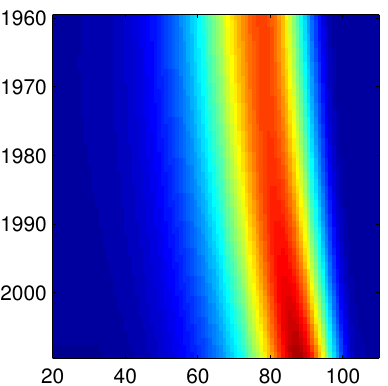}}
\caption{(a) Yearly mortality densities for Luxembourg for the years 1960--2009 (b),(c) Global Fr\'echet regression fits of yearly mortality densities using $X_i = t_i$ and $X_i = (t_i, t_i^2)^T$, respectively. \label{fig: lux}}
\end{figure}

While the quadratic model seems to be indeed  better for both fitting and prediction than the linear model, it still does not capture some aspects of the mortality distributions for Luxembourg, particularly between 1970 and 1980.  Therefore,  local fitting methods will likely prove superior.  Figure~\ref{fig: lux_local} shows the Nadaraya-Watson kernel regression \eqref{eq: NWS} and the local Fr\'echet \eqref{eq: local_frechet_est} fits, using bandwidths $h = 5$ and $h = 7,$ respectively.  These bandwidths were chosen by minimizing the average leave-one-out prediction error over a grid, with minimum values of 0.196 and 0.168, respectively, for Nadaraya-Watson and local Fr\'echet fits.  This represents a 14\% improvement in prediction using the local Fr\'echet fit as compared to the Nadaraya-Watson at the best tuning parameter choices.  From the plot of the absolute differences between these estimates in Figure~\ref{fig: lux_diff}, the superiority  of local Fr\'echet regression for the most part can be attributed to its improved performance near the boundaries. Specifically, the Nadaraya-Watson estimator appears to underestimate the mode of the mortality distribution in the  years preceding 2009.

\begin{figure}[t]
\centering
\subcaptionbox{\label{fig: lux_lc}}[1.5in]{\includegraphics[width = 1.5in, height = 1.5in]{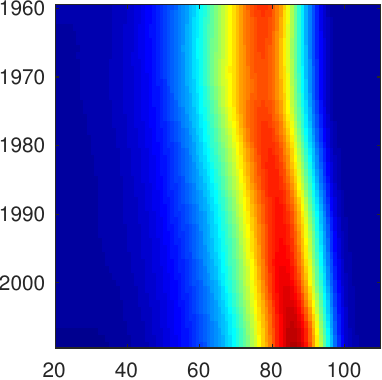}}
\subcaptionbox{\label{fig: lux_lf}}[1.5in]{\includegraphics[width = 1.5in, height = 1.5in]{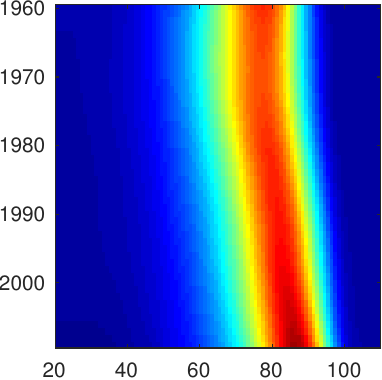}}
\subcaptionbox{\label{fig: lux_diff}}[1.5in]{\includegraphics[width = 1.5in, height = 1.5in]{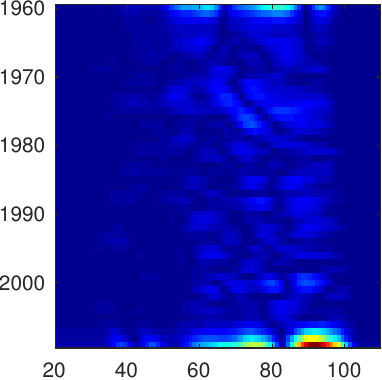}}
\caption{(a), (b) Nadaraya-Watson and local Fr\'echet estimates for Luxembourg (c) Absolute difference between local fits. \label{fig: lux_local}}
\end{figure}

\subsection{Inference and Model Selection}
\label{ss:inf}

Many of the standard inferential tools that are available for ordinary linear regression depend on the algebraic structure of $\R$, and thus are not directly extendable to Fr\'echet regression for metric-valued data.  However, one tool which does generalize is the coefficient of determination, $R^2$.  Recall  that in multiple linear regression modeling with real valued responses, $R^2$ is usually interpreted as the fraction of variance of the response which is explained by a linear relationship with the predictor variables, i.e.
\[
R^2 = 1 - \frac{\Var(Y - \sbn - (\sbo)^T(X-\mu))}{\Var(Y)}.
\]
Using the generalized notions of mean and variance in (\ref{eq: frechet_mv}), we define a  corresponding Fr\'echet $R^2$ coefficient of determination as
\[
R^2_\oplus := 1 - \frac{E\left[d^2(Y, \mp(X))\right]}{V_\oplus}.
\]
Given a random sample $(X_i, Y_i)$, $i = 1, \ldots, n$, $R^2_\oplus$ can be estimated by
\[
\hat{R}^2_\oplus = 1 - \frac{\sn d^2(Y_i, \hmp(X_i))}{\sn d^2(Y_i, \homp)},
\]
where $$\homp = \argmin_{\om \in \Om}n\inv \sn d^2(Y_i, \om)$$ is the sample Fr\'echet mean.  The values $R^2_\oplus$ has similar interpretations as the ordinary coefficient of determination $R^2$ and likely is also useful for  inference and model selection.

In the setting of global Fr\'echet regression, the null hypothesis of no effect is equivalent to testing $\mc{H}_0: R^2_\oplus = 0$, for which the estimate $\hat{R}^2_\oplus$ can be used as a test statistic.  In order to obtain a $p$-value, a permutation test can be performed \cp{lehm:06, higg:04,brad:68}.  First, the values $X_i$, $i = 1, \ldots, n,$ are permuted to form a new sample $\tilde{X}_i$, $i=1,\ldots,n$. For each new predictor sample, a global Fr\'echet regression is fitted,  using the pairs $(\tilde{X}_i, Y_i)$, and the value $\hat{R}^2_\oplus$ is computed for each of these regression fits.  By performing a large number of such permutations and fits,  one then obtains an empirical approximation of the null distribution of the test statistic and a $p$-value by calculating the quantile of the actually observed  $\hat{R}^2_\oplus$ within this null distribution.

Another potential  application of the  coefficient $\hat{R}^2_{\oplus}$ is model selection, where one can mimic the use of the adjusted $R^2$
in linear regression by fitting Fr\'echet regression  models that use various  subsets of the predictor variables.  For a fitted submodel $\mc{M}$ using $q \leq p$ predictor values, the adjusted Fr\'echet $R^2$ is then
\[
\hat{R}^2_{\oplus, \rm{adj}}(\mc{M}) = \hat{R}^2_\oplus - (1 - \hat{R}^2_\oplus)\frac{q}{n - q - 1}.
\]
Let $\mc{C}_q$ be the class of submodels using $q$ predictors, $1 \leq q \leq p$. Computing
\[
q^\ast = \argmax_{1 \leq q \leq p} \max_{\mc{M} \in \mc{C}_q} \hat{R}^2_{\oplus, \rm{adj}}(\mc{M})
\]
the final model can then be taken as $\mc{M}^\ast = \argmax_{\mc{M} \in \mc{C}_{q^\ast}} \hat{R}^2_{\oplus, \rm{adj}}(\mc{M})$.  Another alternative for model selection is to minimize prediction error, which can be estimated by $k$-fold cross validation.

\section{Fr\'echet Regression for Correlation Matrices as Random Objects}
\label{sec: correlation}

\subsection{Computational Details}
\label{ss: comp_corr}

Here we consider a space of random objects $\Om$ which consists of correlation matrices, i.e.,  the space of square $r\times r$ symmetric positive semidefinite matrices with unit diagonal, for some positive integer $r$, and equip $\Om$ with the Frobenius metric $d_F$.  Positive definite matrices have been studied previously from the random object perspective under different metrics \cp{arsi:07, 
pigo:14}.  From a sample $(X_i, Y_i)$, $i = 1,\ldots, n$, the minimization in (\ref{eq: mp_est}) can be reformulated by setting $\hat{B}(x) = n\inv\sn s_{in}(x)Y_i$ and computing (see proof of Proposition ~\ref{prop: correlation}
in the Appendix for details)
\[
\hmp(x) = \argmin_{\om \in \Om} \Fro{\hat{B}(x)}{\om}^2.
\]
Thus, the problem is reduced to finding the correlation matrix which is nearest to the matrix $\hat{B}(x)$.  This problem has been well studied \cp{high:02, qi:06, bors:10}, and in our implementations we used the alternating projections algorithm, written by Nicholas Higham and available at \url{<https://nickhigham.wordpress.com/2013/02/13/the-nearest-correlation-matrix/>}, to carry out this optimization.

\subsection{Functional Connectivity in the Brain}
\label{ss: conn}

In recent years, the problem of identifying functional connectivity between brain voxels or regions has received a great deal of attention, especially for resting state fMRI \cp{alle:14, ferr:13, lee:13, shel:13}.  Subjects are asked to relax while undergoing a fMRI brain scan, where blood-oxygen-level dependent signals are recorded and then processed to yield voxel-specific time courses of signal strength. The connectivity between brain regions is usually quantified by the temporal correlation between representative time signals of the two regions.  Higher levels of correlation are reflective of higher connectivity, giving rise to the question of which subject-specific factors might explain observed variations in connectivity.  When considering $r > 2$ brain regions, the resulting number of pairwise correlations is $r(r-1)/2$, so that standard statistical models are inadequate for investigating the relationship between several predictors and the connectivity response.   Fr\'echet regression can be employed to directly address this issue by viewing the functional connectivity measurements in a natural way as random elements of the space of correlation matrices.

The data for our analysis come from a study of 174 cognitively normal elderly patients, each of whom underwent an fMRI scan at the UC Davis Imaging Research Center.  Preprocessing of the recorded BOLD (blood oxygenation-level-dependent) signals was implemented by adopting the standard procedures of slice-timing correction, head motion correction and normalization, in addition to linear detrending to account for signal drift and band-pass filtering to include only frequencies between 0.01 and 0.08 Hz.

Of particular interest regarding functional connectivity in the resting state is the so-called default-mode network (DMN), including the study of age-related effects 
\cp{ferr:13}.  
In one such study, \ci{meve:13} investigated disruptions between anterior-posterior components in the DMN as subjects age and   
found a decrease in connectivity between a seed region in the left ventral medial prefrontal cortex (lvmPFC) and three other voxels located within the right vmPFC/orbitofrontal (rvmPFC), left ventral posterior cingulate cortex (lvPCC) and right precuneus/PCC (rpPCC) regions. 

To construct a connectivity correlation matrix for each subject, signals at these $r = 4$ locations were extracted and their temporal correlations computed.  These signals are taken over the interval [0, 470] (in seconds), with $T = 236$ measurements available at 2 second intervals.  Hence, for the $i$th subject, the data are in the form of an $T\times r$ signal matrix $S_i$ where the rows correspond to consecutive time points and the columns to distinct voxels.  Define $s_{ijk} = (S_i)_{jk}$ and $\bar{s}_{ik} = T\inv\sum_{j = 1}^T s_{ijk}$.  The connectivity correlation matrix $Y_i$ for the $i$th subject as it is routinely calculated for analyzing connectivity in fMRI has the elements
\[
(Y_i)_{kl} = \frac{\sum_{j = 1}^T(s_{ijk} - \bar{s}_{ik})(s_{ijl} - \bar{s}_{il})}{\left[\left(\sum_{j = 1}^T (s_{ijk} - \bar{s}_{ik})^2\right)\left(\sum_{j = 1}^T (s_{ijl} - \bar{s}_{il})^2\right)\right]^{1/2}}.
\]

\begin{figure}
  \centering
\subcaptionbox{lvmPFC vs. rvmPFC \label{fig: lvmPFC_vs_rvmPFC}}[2.34in]{\includegraphics[scale = 0.9]{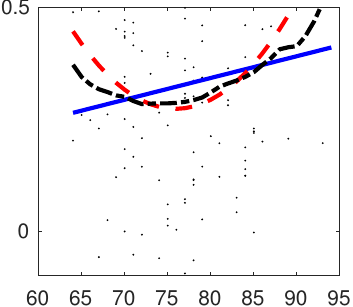}}
\subcaptionbox{lvmPFC vs. lvPCC \label{fig: lvmPFC_vs_lvPCC}}[2.34in]{\includegraphics[scale = 0.9]{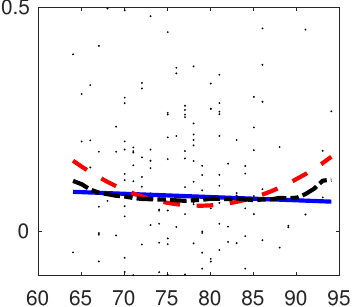}} \\
\subcaptionbox{lvmPFC vs. rpPCC \label{fig: lvmPFC_vs_rpPCC}}[2.34in]{\includegraphics[scale = 0.9]{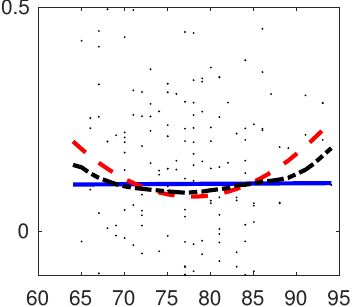}}
\subcaptionbox{rvmPFC vs. lvPCC \label{fig: rvmPFC_vs_lvPCC}}[2.34in]{\includegraphics[scale = 0.9]{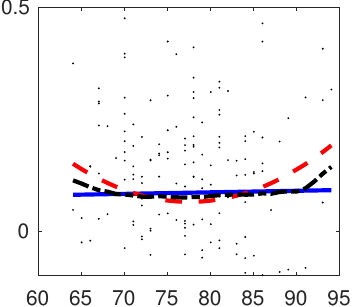}} \\
\subcaptionbox{rvmPFC vs. rpPCC \label{fig: rvmPFC vs rpPCC}}[2.34in]{\includegraphics[scale = 0.9]{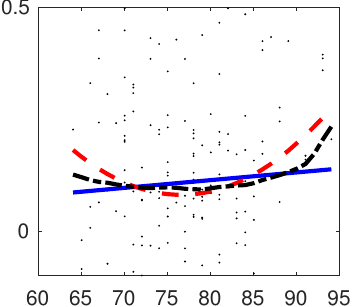}}
\subcaptionbox{lvPCC vs. rpPCC \label{fig: lvPCC_vs_rpPCC}}[2.34in]{\includegraphics[scale = 0.9]{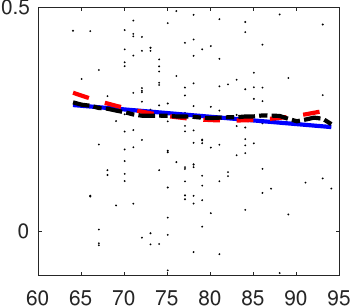}}
\caption{Component-wise scatterplots and Fr\'echet regression fits for voxel-to-voxel correlation as a function of age.  The linear, quadratic, and Nadaraya-Watson ($h = 7$) fits are represented by solid, dashed, and dot-dashed lines, respectively. Note that these fits are derived from  Fr\'echet regression analysis where entire correlation matrices serve as random object responses. \label{fig: dmn_con}}
\end{figure}

In our regression model, we use age as a predictor of connectivity and fit both linear and quadratic models, i.e. $X_i = Z_i$ and $X_i = (Z_i, Z_i^2)^T$, where $Z_i$ is the age of subject $i$, $i = 1, \ldots, 174$.  

Since it is unclear whether the global Fr\'echet regression model \eqref{eq: mp} holds, we also fit the regression nonparametrically using the Nadaraya-Watson smoother \eqref{eq: NWS} over a range of bandwidths.  One notable difference between the current data and those used in \ci{meve:13} is the age range.  The current analysis includes only elderly subjects, aged 64 to 94 years, while \ci{meve:13} included subjects between 19 and 80 years of age.  It has been observed previously \cp{onod:12, ferr:13} that age-related effects are more difficult to detect in later years.  Thus, the goal for our analysis is to investigate if the decreases in connectivity observed in \ci{meve:13} are also found among a group of strictly elderly subjects, or whether the pattern is  different.

For each regression fit, the estimated mean-square prediction error (MSPE) was calculated using five-fold cross validation, averaged over 50 runs.  The MSPE values for the linear and quadratic models were $0.6489$ and $0.6386,$ respectively.  For the Nadaraya-Watson fit, the minimum MSPE over a grid of bandwidths was $0.6393$, for bandwidth $h = 7.$  The linear model had a $p$-value of $0.58$ and $\hat{R}_\oplus^2 = 0.0041$, while the quadratic model was a much better fit, with a $p$-value of $0.014$ and $\hat{R}_\oplus^2 = 0.0288$.  Figure~\ref{fig: dmn_con} illustrates the regression fits for each component of the lower subdiagonal of the correlation matrix.  The visual and numerical results suggest that the quadratic global Fr\'echet regression model is adequate for these data, as the fit is quite similar to the Nadaraya-Watson estimator without requiring any bandwidth selection.  Thus, age-related changes in connectivity seem to be more subtle in later years, with subjects over 85 demonstrating greater connectivity between some regions than younger subjects between the ages of 75 and 85.  While some studies have found increased connectivity with age \cp{ferr:13}, the quadratic model reveals that simple linear associations between age and connectivity may be inadequate.

\section{Local Fr\'echet Regression for Spherical Data}
\label{sec: sphere}

As a final illustration, we implement  local Fr\'echet regression for a situation where the random object responses lie in a Riemannian manifold object space.  Specifically, choose $\Om = S^2$ as the unit sphere in $\R^3$, with geodesic distance $d(y, z) = \arccos(z^Ty)$ and consider the regression function
\[
\mp(x) = ((1 - x^2)^{1/2}\cos(\pi x), (1 - x^2)^{1/2}\sin(\pi x), x), \quad x \in (0, 1),
\]
which maps a spiral on the sphere.  To generate a random sample $(X_i, Y_i)$, $i = 1, \ldots, n$, $X_i \sim \mathcal{U}(0, 1)$ was first sampled, followed by a bivariate normal random vector $U_i$ on the tangent space $T_{\mp(X_i)}\Om$.  Finally, with $\lVert \cdot \rVert_E$ being the Euclidean norm,
\[
Y_i = \textrm{Exp}_{\mp(X_i)}(U_i) = \cos(\lVert U_i\rVert_E)\mp(X_i) + \sin(\lVert U_i \rVert_E)\frac{U_i}{\lVert U_i \rVert_E}.
\]
Random samples of size $n = 50, 100, 200$ were generated under two noise scenarios, with 200 runs for each simulation.  In both noise scenarios, the components of $U_i$ were independent, with each having a variance of $0.2$ and $0.35$ in the low and high noise scenarios, respectively.  Figure~\ref{fig: sphere_samples} shows two sample data sets of size 50 for the two noise scenarios.

\begin{figure}
  \centering
\subcaptionbox{Low Noise, $n = 50$}[2.4in]{\includegraphics[scale = 0.45]{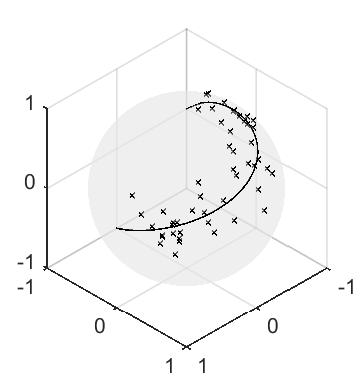}}
\subcaptionbox{High Noise, $n = 50$}[2.4in]{\includegraphics[scale = 0.45]{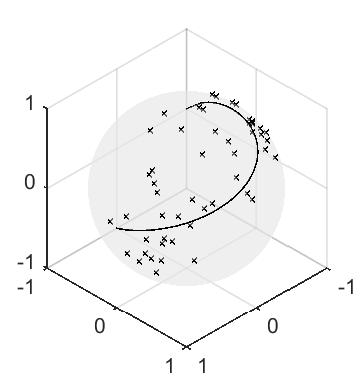}}
\caption{Sample simulation data sets of size $n = 50$ under low (left) and high (right) noise settings.  The true regression curve is shown by the solid line. \label{fig: sphere_samples}}
\end{figure}

For estimation, a grid of bandwidths $h \in (0.05, 0.3)$ was used for the smoothing, with $K$ being the Epanechnikov kernel;  this estimation was performed for both local Fr\'echet regression and the Nadaraya-Watson smoother.  The necessary optimization was performed using the trust regions algorithm as implemented in the ManOpt toolbox for Matlab \cp{boum:14}.  While we found this to be an adequate computational tool for our simulations, it may be necessary in some scenarios to implement a stochastic optimization scheme, such as the annealing algorithm \cp{yuan:12}.  We also implemented an alternative regression method for response data on a nonlinear manifold via smoothing splines \cp{su:12}, with code provided by one of the authors.

To compare local Fr\'echet regression with Nadaraya-Watson and spline smoothing, for each combination of noise setting and sample size, the mean integrated squared error (MISE) of each method was computed across a range of tuning parameters.  For our method and Nadaraya-Watson smoothing, this was done over the grid of bandwidths.  For the spline method, the three parameters and their values used for each simulation were $T = 50t + 1$, $t = 1,5,10,20$, $\epsilon = 10^{-l}$, $l = 2,\ldots,6$ and $\lambda = 10^k$, $k = -5,\ldots,5.$  The minimum MISE values are shown in Table~\ref{tab: sphere_mise}.  We see that local Fr\'echet regression outperforms the other methods in every setting, while the spherical spline method is not a close competitor.  Additionally, one can get a sense of the bias of the two Fr\'echet estimation techniques by taking Fr\'echet averages of the fits $\hmp(x)$ across simulations, for a grid of levels $x \in [0,1]$.  For example, these averaged local Fr\'echet and Nadaraya-Watson regression fits, using the bandwidths which minimize MISE, are shown in Figure~\ref{fig: sphere_fits} for the low noise setting with $n = 100$.  Again, the local Fr\'echet method is found to be superior, especially in terms of performance near the boundaries.

\begin{table}[h!]
\large
  \centering
  \caption{Best MISE values (multiplied by 100 for clarity) for local Fr\'echet regression (LF), Nadaraya-Watson (NW) and spherical spline (SS) fits.  In parentheses, the minimizing bandwidths $h$ are given for the first two methods, while the minimizing triples $(T, \epsilon,\lambda)$ are given for the spline method.\label{tab: sphere_mise}}
  \begin{tabular}{ll|c|c|c|}
  \hline
    Noise & $n$ & NW & LF & SS \\ \hline
    \multirow{3}{*}{Low} & 50 & 1.34(0.13) & 0.97(0.22) & 5.47 (51, 0.01, 0.01) \\
    & 100 & 0.74(0.13) & 0.51(0.19) & 5.46 (51, 0.01, 1000)\\
    & 200 & 0.45(0.09) & 0.31(0.15) & 5.42 (51,  0.01, 100) \\
    \hline
    \multirow{3}{*}{High} & 50 & 3.00(0.19) & 2.61(0.34) & 16.99 (251, 0.01, 1)\\
    & 100 & 1.73(0.16) & 1.41(0.26) & 16.26 (51, 0.001, 0.00001)\\
    & 200 & 0.99(0.13) & 0.76(0.21) & 13.32 (51, 0.001, 0.00001)
  \end{tabular}
\end{table}

\begin{figure}[h!]
  \centering
\subcaptionbox{Local Fr\'echet regression}[2.4in]{\includegraphics[scale = 0.45]{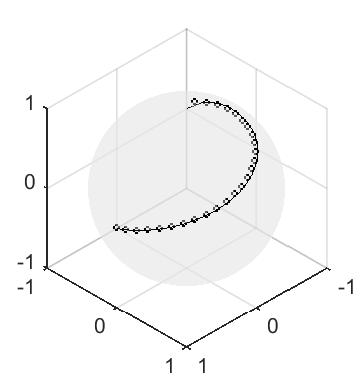}}
\subcaptionbox{Nadaraya-Watson smoothing}[2.4in]{\includegraphics[scale = 0.45]{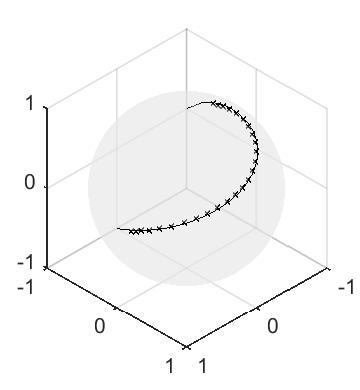}}
\caption{Fr\'echet-averaged regression curve fits for local Fr\'echet regression (left, circle markers) and Nadaraya-Watson smoothing (right, `x' markers), with true regression curve given for reference (solid).  These are from the low noise simulation with $n = 100$. \label{fig: sphere_fits}}
\end{figure}

\section{Discussion}
\label{sec: dis}

The proposed global and local Fr\'echet regression models 
are new tools for the analysis of random objects that are increasingly encountered in modern data analysis. They extend the fundamental notion of  a Fr\'echet mean to that of conditional Fr\'echet means.  We provide theoretical justifications including rates of convergence  for both global and local versions.  
The 
pointwise 
rates of convergence are optimal for both global and local versions in the sense that in the special case of Euclidean objects 
they correspond to the known optimal rates.  For objects in general metric spaces, the
entropy-based rates additionally depend on the complexity of the space through the
parameter $\gamma$ in (P1) and (U1), and coincide with the Euclidean rates when
$\gamma = 0$, as is the case for correlation matrices and finite-dimensional Riemannian
manifolds.  For infinite-dimensional spaces with $\gamma > 0$, additional structure can
sometimes be exploited to recover the Euclidean rates, as we demonstrate for the
Wasserstein space in Section~\ref{ss: wass_hilbert}.

For practical applications of the global Fr\'echet regression model, we introduce the concept of the Fr\'echet coefficient of determination, $R_{\oplus}^2$, and explore its potential use for testing.  We focus in this paper on estimation, and future work will be needed to develop formal tests, confidence sets and predictor selection.
For the development of the local version of Fr\'echet regression it proved  necessary to revisit what is meant by the concept of a local regression and to clarify the nature of the target. In data examples, local Fr\'echet regression proved competitive with previously discussed local smoothing methods for special object spaces. An interesting special case for which we obtain limit distributions is the case of responses that live in a Hilbert space, such as functional data. Indeed, as pointed out by a referee, this model may prove useful in the case of responses which lie on a Hilbert manifold as an extrinsic regression technique in infinite dimensions (see Chapters 11 and 18 of \ci{patr:15}).

Conditional Fr\'echet means and the associated regression approaches have a wide range of applications that include responses that lie in a Riemannian manifold as a special case. For this case we show that our general and straightforward approach is not only theoretically competitive but also works well in simulations.  In this and other situations, uniqueness of the Fr\'echet mean is sometimes not guaranteed, e.g., in the  case of a uniform distribution on the sphere, in contrast to other cases that we explored, where it is unique.  For manifolds, it is often assumed that $\Om$ is complete in order to prove existence of a Fr\'echet mean \cp{bhat:03}.  Recently, \ci{le:14} showed that the cut locus of a minimizer of the Fr\'echet function necessarily has probability zero, lending further insight into the distributional limitations which allow for existence and uniqueness of Fr\'echet means.  When conditional Fr\'echet means are not unique, one may need to deal with sets of Fr\'echet means that consist of many elements \cp{patr:15,ziez:77}.
Extensions that fall within the framework that we outline also include special types of linear models such as analysis of variance and, more generally, regression models that include indicators among the predictors, as well as polynomial regression models or models with interactions. 

\appendix

\section{Proofs of Theoretical Results}

\subsection{Propositions~\ref{prop: wass}--\ref{prop: manifold}}

\begin{Proposition}
  \label{prop: wass}
  The space $(\Om, d_W)$ defined in Example~\ref{exm: wass} satisfies assumptions (P0)--(P2) and (U0)--(U2), with $\gamma = 1/2$ in (P1) and (U1), $\beta = 2$ in (P2), and $\alpha = 2$ in (U2).
\end{Proposition}
\begin{proof}
For any distribution $\om \in \Om$, let $Q(\om)$ be the corresponding quantile function.  Similarly, $Q\inv(h) \in \Om$ is the distribution corresponding to $h \in Q(\Om)$.  Let $\Ltwoip{\cdot}{\cdot}$, $\Ltwonorm{\cdot}$ and $\Ltwo{\cdot}{\cdot}$ be the $L^2$ inner product, norm and distance on $[0,1]$, respectively.  Since $E\left(|s(X, x)|\Ltwonorm{Q(Y)}\right)$ is finite, the Riesz Representation Theorem implies the existence of an element $g_x \in L^2[0,1]$ such that
\[
E(s(X, x)\Ltwoip{Q(Y)}{h}) = \Ltwoip{g_x}{h}
\]
for all $h \in L^2[0,1]$.  Define $\hat{g}_x = n\inv\sn s_{in}(x) Q(Y_i)$.  Then properties of the $L^2$ distance imply
\begin{align*}
M(\om, x) &= E\left(s(X, x)\Ltwo{Q(Y)}{g_x}^2\right) + \Ltwo{Q(\om)}{g_x}^2, \\
M_n(\om, x) &= \frac{1}{n}\sum_{i = 1}^n s_{in}(x)\Ltwo{Q(Y_i)}{\hat{g}_x}^2 + \Ltwo{Q(\om)}{\hat{g}_x}^2,
\end{align*}
yielding the solutions
\[
\mp(x) = Q\inv\left(\argmin_{h \in Q(\Om)} \Ltwo{h}{g_x}^2\right), \quad \hmp(x) = Q\inv\left(\argmin_{h \in Q(\Om)} \Ltwo{h}{\hat{g}_x}^2\right),
\]
which exist and are unique by convexity of $Q(\Om)$ for any $x \in \R^p$, hence proving (P0) and (U0).  Additionally, $\mp(x)$ is characterized by
\[
\Ltwoip{g_x - Q(\mp(x))}{h - Q(\mp(x))} \leq 0
\]
for all $h \in Q(\Om)$.  Consequently, we may take $C = D = 1$, $\beta = \alpha = 2$ and $\eta$ and $\tau$ arbitrary in (P2) and (U2).

Lastly, we show that (U1) holds with $\gamma = 1/2$, which of course implies (P1) with $\gamma = 1/2$.  
Recall that $\Om$ is restricted to distributions with support on a common compact interval $\mathcal{I}$, so that $Q(\Om) \subset L^2[0,1]$ is uniformly bounded.  Let $\mc{Q} = \{Q_\om: \om \in \Om\}$ and denote by $B^2_\delta(Q)$ the $L^2$ ball of radius $\delta$ centered at $Q \in \mc{Q}$, and by $B^2_\delta(Q; \mc{Q})$ the same ball restricted to $\mc{Q}$.  The relationship between covering numbers and bracketing numbers in function spaces and Theorem 2.7.5 of \ci{well:96} imply that
\[
\begin{split}
\sup_{\om \in \Om} \log N(\delta\epsilon, B_\delta(\om), \dw) &= \sup_{\om \in \Om} \log N(\delta\epsilon, B_\delta^2(Q_\om; \mc{Q}), d_{L^2}) \\
&= \sup_{Q \in \mc{Q}} \log N(\delta\epsilon, B_\delta^2(Q; \mc{Q}), d_{L^2}) \\
&\leq \sup_{Q \in \mc{Q}} \log N_{[\,]}(2\delta\epsilon, B_\delta^2(Q; \mc{Q}), d_{L^2}) \\
&\leq \log N_{[\,]}(2\delta\epsilon, \mc{Q}, d_{L^2}) \leq \frac{K}{\delta\epsilon}
\end{split}
\]
for a constant $K$ depending only on the length of $\mathcal{I}$.  Since $m_\oplus(x) \in \Om$ for every $x$, this bound applies in particular with $\om = m_\oplus(x)$, uniformly over $x$, so that the integral in (U1) is bounded by
$$
\int_0^1\sqrt{1 + \frac{K}{\delta\epsilon}}\;d \epsilon \leq \int_0^1\left(1 + \sqrt{\frac{K}{\delta}}\epsilon^{-1/2}\right) d\epsilon = 1 + 2\sqrt{\frac{K}{\delta}} = O(\delta^{-1/2})
$$
as $\delta \rightarrow 0$, establishing (U1), and hence (P1), with $\gamma = 1/2$.
\end{proof}

\begin{Proposition}
  \label{prop: correlation}
  The space $(\Om, d_F)$ defined in Example~\ref{exm: correlation} satisfies assumptions (P0)--(P2) and (U0)--(U2), with $\gamma = 0$ in (P1) and (U1), $\beta = 2$ in (P2), and $\alpha = 2$ in (U2).
\end{Proposition}

\begin{proof}
Here, $Y$ is an $r\times r$ correlation matrix.  Denote the elements of $Y$ as $Y(j, k)$, $1\leq j,k \leq r$.  Let $\Froip{\cdot}{\cdot}$, $\Fronorm{\cdot}$ and $\Fro{\cdot}{\cdot}$ be the Frobenius inner product, norm and distance, respectively.  Let $B_{jk}(x) = E\left(s(X, x)Y(j, k)\right)$ and $\hat{B}_{jk}(x) = n\inv\sn s_{in}(x)Y_i(j, k)$.  Then properties of the Frobenius distance imply that
\begin{align*}
M(\om, x) &= M(B(x), x) + \Fro{\om}{B(x)}^2, \\
M_n(\om, x) &= M_n(\hat{B}(x), x) + \Fro{\om}{\hat{B}(x)}^2,
\end{align*}
yielding the solutions
\[
\mp(x) = \argmin_{\om \in \Om} \Fro{\om}{B(x)}^2, \quad \hmp(x) = \argmin_{\om \in \Om} \Fro{\om}{\hat{B}(x)}^2,
\]
which exist and are unique by the convexity of $\Om$ for any $x \in \R^p$, hence proving (P0) and (U0).  Additionally, $\mp(x)$ is characterized by
\[
\Froip{B(x) - \mp(x)}{\om - \mp(x)} \leq 0
\]
for all $\om \in \Om$.  Consequently, we may take $\eta$ and $\tau$ arbitrary, $C = D = 1$ and $\beta = \alpha = 2$ in (P2) and (U2).

Lastly, since $\Om$ is a bounded subset of the larger finite-dimensional Euclidean space of $r\times r$ matrices, for any $\omega \in \Omega,$
$$
N(\delta\epsilon, B_\delta(\omega), d_F) \leq N(\epsilon/2, \overline{B}_1(\omega), d_F) \leq K\epsilon^{-r^2}
$$
where $\overline{B}_1(\omega)$ denotes the ball of radius 1 centered at $\omega$ in the space of all $r\times r$ matrices and $\epsilon \in (0,1]$, by a standard volume argument in $\R^{r^2}$, where $K > 1$ depends on $r$ only.  Thus, the integral in (U1) is bounded by
\begin{align*}
\int_0^1 \sqrt{1 + \log K - r^2\log\epsilon}\;d \epsilon &\leq 1 + \sqrt{\log K} + r\int_0^1\sqrt{-\log \epsilon}\;d \epsilon \\
&= 1 + \sqrt{\log K} +r\int_0^\infty e^{-y}\sqrt{y}\; dy < \infty,
\end{align*}
using the substitution $y = -\log \epsilon.$  Since this bound does not depend on $\delta,$ (U1) holds, with $\gamma = 0$, and thus (P1) as well.

\end{proof}

\begin{Proposition}
  \label{prop: manifold}
  The space $(\Om, d)$ defined in Example~\ref{exm: manifold} satisfies (P1) and (U1) with $\gamma = 0$,
  provided the Riemannian metric is equivalent to the ambient Euclidean metric.
   Let $T_{\om}\Om$ be the tangent bundle at $\om$ and ${\rm Exp}_{\om}$ and ${\rm Log}_{\om}$ be the exponential and logarithmic manifold maps at $\om$.  For $u \in T_{\om}\Om$, define
  \[
  g_{\om}(u) = M\left({\rm Exp}_{\om}(u), x\right), \quad h_{\om}(u) = M_n\left({\rm Exp}_{\om}(u), x\right).
  \]
  If (P0) holds and $g_{\mp(x)}''(0)$ is positive definite, then (P2) holds with $\beta = 2$.  Similarly, if (U0) holds then
  \[
  \inf_{\lVert x\rVert_E \leq B} \lambda_{{\rm min}}(g_{\mp(x)}''(0)) > 0 \quad
  \]
  implies (U2) with $\alpha = 2$, where $\lambda_{{\rm min}}(A)$ is the smallest eigenvalue of a square matrix $A$.
\end{Proposition}

\begin{proof}
  Since $\Om$ is bounded and of finite dimension, (U1) follows by an argument similar to the last part of the previous proof due to metric equivalency, whence (P1) also follows, with $\gamma = 0$ in both.  
  If (P0) holds, let $\eps$ be the injectivity radius at $\mp(x)$ and consider $\om$ such that $d(\om, \mp(x)) < \eps$. Taking $u_x = {\rm Log}_{\mp(x)}(\om)$,
  \[
  M(\om, x) - M(\mp(x), x) = g_{\mp(x)}(u) - g_{\mp(x)}(0) = \frac{1}{2}u_x^Tg_{\mp(x)}''(u_x^\ast)u_x
  \]
  for some $u_x^{\ast}$ between $0$ and $u_x$.  Since $u_x^Tu_x = d^2(\om, \mp(x))$ and $g_{\mp(x)}$ is continuous, the condition on $g''_{\mp(x)}(0)$ implies (P2) with $\beta = 2$.  Similar arguments using the other conditions show that $\alpha = 2$ in (U2) is permissible.
\end{proof}

\subsection{Proofs of results in Section~\ref{sec: theory}}

Throughout, the symbol $\rightsquigarrow$ will denote weak convergence and the notation $l^\infty(\Om)$ denotes the space of bounded functions on $\Om$.  The ordinary Euclidean norm on $\R^p$ will be denoted by $\lVert \cdot \rVert_E$ and the Frobenius norm by $\lVert \cdot \rVert_F$.  For simplicity of notation, when $x$ is fixed, the dependence of objects such as $M$, $\mp$, etc. on $x$ will be dropped.

\begin{proof}[Proof of Theorem~\ref{lma: con}]
We first consider fixed $x \in \R^p$.  By Corollary 3.2.3 in \ci{well:96}, convergence of $\sup_{\om \in \Om} |M_n(\om) - M(\om)|$ to zero in probability is sufficient.  To do this, we show $M_n \rightsquigarrow M$ in $l^\infty(\Om)$ and apply 1.3.6 of \ci{well:96}.  This weak convergence is proved (see Theorem 1.5.4 of \ci{well:96}) by showing that
\begin{enumerate}[i)]
\item $M_n(\om) - M(\om) = o_p(1)$ for all $\om \in \Om$ and
\item $M_n$ is asymptotically equicontinuous in probability, i.e. for all $\eps$, $\eta > 0$, there exists $\delta > 0$ such that
\[
\limsup_n P\left(\sup_{d(\om_1, \om_2) < \delta} |M_n(\om_1) - M_n(\om_2)| > \eps\right) < \eta.
\]
\end{enumerate}

Begin with i).  Set
\begin{equation}
\label{eq: si_weights}
s_i = \left[1 + (X_i - \mu)^T\Si\inv(x - \mu)\right]
\end{equation}
and define
\[
\tilde{M}_n(\om) = n\inv\sn s_id^2(Y_i, \om).
\]
Then for all $\om \in \Om$, $E(\tilde{M}_n(\om)) = M(\om)$ and
\[
\Var(\tilde{M}_n(\om)) \leq n\inv\diam^2(\Om)E(s_i^2) \leq 2n\inv\diam^2(\Om)(1 + (x-\mu)^T\Si\inv(x-\mu)),
\]
so $\tilde{M}_n(\om) - M(\om) = o_p(1)$.  Also, setting
\begin{align}
  W_{0n}:= W_{0n}(x) &= \xbar\Si\inv(x - \xbar) - \mu^T\Si\inv(x - \mu), \label{eq: Ws}\\
  W_{1n}:= W_{1n}(x) &= \Si\inv(x - \mu) - \hSi\inv(x - \xbar),\nonumber
\end{align}
we have $s_{in} - s_i = W_{0n} + W_{1n}^TX_i$. Then
\[
M_n(\om) - \tilde{M}_n(\om) = \frac{W_{0n}}{n}\sn d^2(Y_i, \om) + \frac{W_{1n}^T}{n}\sn X_id^2(Y_i, \om) = o_p(1)
\]
for all $\om \in \Om$, since $W_{0n}$ and $\lVert W_{1n}\rVert_E$ are both $O_p(n^{-1/2})$.  Using the triangle inequality, we have proven i).  Hence, for any $k \in \mathcal{N}$ and $\om_1,\ldots,\om_k \in \Om$, we have $(M_n(\om_1), \ldots, M_n(\om_k)) \rightsquigarrow (M(\om_1), \ldots, M(\om_k))$.

Moving on to ii), for any $\gamma_1$, $\gamma_2 \in \Om$,
\begin{align*}
  |M_n(\gamma_1) - M_n(\gamma_2)| &\leq \frac{1}{n}\sn|s_{in}||d(\om_i, \gamma_1) - d(\om_i, \gamma_2)||d(\om_i, \gamma_1) + d(\om_i, \gamma_2)| \\
  &\leq 2\diam(\Om)d(\gamma_1, \gamma_2)\left(\frac{1}{n}\sn |s_i + W_{0n} + W_{1n}^TX_i|\right) \\
  &=O_p(d(\gamma_1, \gamma_2)),
\end{align*}
where the $O_p$ term is independent of $\gamma_1$ and $\gamma_2$.  Hence, $$\sup_{d(\om_1,\om_2)<\delta} |M_n(\om_1) - M_n(\om_2)| = O_p(\delta),$$ which proves ii).  This shows that $d(\mp(x), \hmp(x)) = o_p(1)$.

For the uniform result, consider the process $Z_n(x) = d(\hmp(x), \mp(x))$, so $Z_n(x) = o_p(1)$ for any $x \in \R^p$.  By Theorem 1.5.4 in \ci{well:96}, it suffices to show that, for any $S > 0$ and as $\delta \ra 0$,
\[
\limsup_{n \ra \infty}\quad P\left(\sup_{\substack{\lVert x - y \rVert_E < \delta \\ \lVert x \rVert_E, \lVert y \rVert_E \leq B}} |Z_n(x) - Z_n(y)| > 2S\right) \ra 0.
\]
Because $|Z_n(x) - Z_n(y)| \leq d(\mp(x),\mp(y)) + d(\hmp(x),\hmp(y))$, it suffices to show that $\mp(\cdot)$ is uniformly continuous for $\lVert x \rVert_E \leq B$ and that, as $\delta \ra 0$,
\begin{equation}
\label{eq: hmp_cont}
\limsup_{n\ra\infty} P\left(\sup_{\substack{\lVert x - y \rVert_E < \delta \\ \lVert x \rVert_E, \lVert y \rVert_E \leq B}} d(\hmp(x),\hmp(y)) > S\right) \ra 0.
\end{equation}

Let $\delta > 0$ and $x$, $y \in \R^p$ with $\lVert x - y\rVert_E < \delta$.  From the form of $M$, it is clear that $\sup_{\om \in \Om}|M(\om, x) - M(\om, y)| \ra 0$ as $\delta \ra 0$.  Assumption (U0) then implies that $\mp$ is continuous at $x$, and thus uniformly continuous over $\lVert x \rVert_E \leq B$.  To show (\ref{eq: hmp_cont}), let $\eps > 0$ and suppose $d(\hmp(x), \hmp(y)) > \eps$ with $\lVert x\rVert_E$, $\lVert y\rVert_E \leq B$.  Then (U0) and the form of $M_n$ imply that
\[
\zeta \leq \sup_{\substack{\lVert x - y \rVert_E < \delta \\ \lVert x \rVert_E, \lVert y \rVert_E \leq B}} \sup_{\om \in \Om}|M_n(\om, x) - M_n(\om, y)| = O_p(\delta),
\]
and the result follows when $\delta \ra 0$.
\end{proof}

\begin{proof}[Proof of Theorem~\ref{thm: rate}]
Let $x \in \R^p$ being fixed and write $\mp = \mp(x)$.  We follow the proof of Theorem 3.2.5 in \ci{well:96} with a few modifications.  A key component of this proof is the process $V_n(\om) = M_n(\om) - M(\om)$.  Let $D_i(\om) = d^2(Y_i, \om) - d^2(Y_i, \mp)$ and $s_i$ be as in (\ref{eq: si_weights}).  Then
\begin{align}
|V_n(\om) - V_n(\mp)| &\leq \left|\frac{1}{n}\sn(s_{in} - s_i)D_i(\om)\right| \label{eq: V_expand}\\
&\hspace{1cm} + \left|\frac{1}{n}\sn\left(s_iD_i(\om) - E(s_iD_i(\om))\right)\right|. \nonumber
\end{align}
This quantity needs to be controlled for small $d(\om, \mp)$.  First, let $W_{0n}$ and $W_{1n}$ be as defined in (\ref{eq: Ws}).  To control the first term on the right-hand side of (\ref{eq: V_expand}), observe that
\[
\sup_{d(\om, \mp) < \delta} \left|\frac{1}{n}\sn(s_{in} - s_i)D_i(\om, x)\right| \leq \frac{2\diam(\Om)\delta}{n}\sum_{i = 1}^n |W_{0n}(x) + W_{1n}(x)^TX_i|,
\]
so that the left hand side is $O_p(\delta n^{-1/2})$.  Using this fact, we can define $$B_R = \left\{\sup_{d(\om, \mp) < \delta} \left|\frac{1}{n} \sn \left(s_{in} - s_i\right)D_i(\om,x)\right| \leq R\delta n^{-1/2}\right\}$$ for $R > 0$, so that $P(B_R^c) \ra 0$ as $R \ra \infty$.

Next, to control the second term on the right-hand side of (\ref{eq: V_expand}) uniformly over small $d(\om, \mp)$, define the functions $g_\om: \R^p\times\Om \ra \R$ as
\[
g_\om(z, y) = \left[1 + (z - \mu)^T\Si\inv(x - \mu)\right]d^2(y, \om)
\]
and the function class
\[
\mc{M}_\delta := \{g_\om - g_{\mp}:\; d(\om, \mp) < \delta\}.
\]
An envelope function for $\mc{M}_\delta$ is $G_\delta(z) = 2\diam(\Om)\delta|1 + (z - \mu)^T\Si\inv(x - \mu)^T|$, and $E(G_\delta(X)^2) = O(\delta^2)$.  Define $J = J(\delta)$ to be the entropy integral given in (P1), so that $J = O(\delta^{-\gamma})$ as $\delta \rightarrow 0.$  Then Theorems 2.7.11 and 2.14.2 of \ci{well:96} and (P1) imply that, for small enough $\delta$,
\begin{equation}
\label{eq: diff_bound2}
E\left(\sup_{d(\om, \mp) < \delta} \left|\frac{1}{n}\sn\left(s_iD_i(\om, x) - E(s_iD_i(\om, x))\right)\right| \right) \leq \frac{J \left[E(G_\delta(X)^2)\right]^{1/2}}{\sqrt{n}},
\end{equation}
so that the left-hand side is $O(\delta^{1-\gamma} n^{-1/2})$.  Hence, combining (\ref{eq: V_expand}), (\ref{eq: diff_bound2}) and the definition of $B_R$, for small $\delta$,
\[
E\left(I_{B_R} \sup_{d(\om, \mp) < \delta}|V_n(\om) - V_n(\mp)|\right) \leq a\delta^{1-\gamma} n^{-1/2},
\]
for some $a > 0$.

To finish, set $r_n = n^{\frac{\beta}{4(\beta+\gamma-1)}}$ and $$S_{j,n}(x) = \{\om: 2^{j-1} < r_nd(\om,\mp(x))^{\beta/2} \leq 2^j\}.$$ Choose $\eta > 0$ to satisfy (P2) and also small enough that (P1) holds for all $\delta < \eta$ and set $\tilde{\eta} := \eta^{\beta/2}$.  For any integer $L$,
\begin{align}
&P\left(r_nd(\hmp, \mp)^{\beta/2} > 2^L\right) \leq P(B_R^c) + P(2d(\hmp, \mp) \geq \eta) \label{eq: bound} \\
&\hspace{1cm} + \sum_{\substack{j\geq L\\ 2^j \leq r_n\tilde{\eta}}}P\left(\left\{\sup_{\om \in S_{j,n}}|V_n(\om) - V_n(\mp)| \geq C\frac{2^{2(j-1)}}{r_n^2}\right\}\cap B_R\right),\nonumber
\end{align}
where $P(B_R^c) \ra 0$ as discussed previously and the second term goes to zero by Lemma~\ref{lma: con}.  For each $j$ in the sum on the right-hand side of (\ref{eq: bound}), we have \mbox{$d(\om, \mp) \leq \left(\frac{2^j}{r_n}\right)^{2/\beta} \leq \eta$}, so this sum is bounded by
\[
4aC\inv\sum _{\substack{j\geq L\\ 2^j \leq r_n\tilde{\eta}}} \frac{2^{2j(1-\beta-\gamma)/\beta}}{r_n^{2(1-\beta-\gamma)/\beta}\sqrt{n}} \leq 4aC\inv \sum_{j \geq L} \left(\frac{1}{4^{(\beta+\gamma-1)/\beta}}\right)^j.
\]
Because $\beta + \gamma> 1$, the last series converges and hence this probability can be made small by choosing $L$ large.  This proves the desired result that $d(\hmp, \mp) = O_p\left(r_n^{-2/\beta}\right) = O_p\left(n^{-\frac{1}{2(\beta + \gamma - 1)}}\right)$.

For the uniform result over $\lVert x \rVert_E \leq B$, use the fact that $W_{0n}(x)$ and $\lVert W_{1n}(x) \rVert_E$ are both $O_p(n^{-1/2})$, uniformly over $\lVert x \rVert_E \leq B$. Then
\begin{equation}
\label{eq: diff_bound3}
\sup_{\lVert x \rVert_E \leq B} \sup_{d(\om, \mp(x)) < \delta} \left|\frac{1}{n}\sn (s_{in}(x) - s_i(x))D_i(\om, x)\right| = O_p(\delta n^{-1/2}).
\end{equation}
Then define $$A_R = \left\{\sup_{\lVert x \rVert_E \leq B}\sup_{d(\om, \mp(x)) < \delta} \left|\frac{1}{n} \sn \left(s_{in}(x) - s_i(x)\right)D_i(\om, x)\right| \leq R\delta n^{-1/2}\right\}$$ for $R > 0$, so $P(A_R^c) \ra 0$. Using the definition of $s_i(x)$ in (\ref{eq: si_weights}), we can bound the second term on the right-hand side of (\ref{eq: V_expand}) by
\begin{align*}
&\lVert\Si\inv(x - \mu)\rVert_E \sum_{j = 1}^p\left|\frac{1}{n}\sn(X_{ij} - \mu_j)D_i(\om,x) - E((X_{ij} - \mu_j)D_i(\om,x))\right|\\
&\hspace{1.5cm} + \left|\frac{1}{n}\sn\left[D_i(\om, x) - E(D_i(\om, x))\right]\right|.
\end{align*}

For $\delta > 0,$ set $U_\delta = \{(x, \omega): \lVert x\rVert_E \leq B, d(m_\oplus(x), \omega) < \delta\}$.  
Next, set
\[
h^j_{x,\om}(z, y) = \begin{cases} d^2(m_\oplus(x),y) - d^2(\om, y), & j = 0, \\ (z_j - \mu_j)\left[d^2(m_\oplus(x),y) - d^2(\om, y)\right], &  j = 1,\ldots p,\end{cases}
\]
and define the classes of functions $\mc{N}^j_\delta = \{h^j_{x, \om}: (x,\om) \in U_\delta\}.$
Assumption (U2) can be used to show, for small $\lVert x_1 - x_2 \rVert_E,$ there is $L_B > 1$ such that $$d(m_\oplus(x_1),m_\oplus(x_2)) \leq L_B\lVert x_1 - x_2 \rVert_E^{1/(\alpha - 1)}.$$  Then $\mc{N}^j_\delta$ are Lipschitz classes for small $\delta$ in the sense that
\[
|h^j_{x_1,\om_1}(z,y) - h^j_{x_2,\om_2}(z,y)| \leq C_j\left[\lVert x_1 - x_2\rVert_E^{1/(\alpha - 1)} + d(\omega_1,\omega_2)\right],
\]
where $C_j = 2L_B\diam(\Om)$ if $j = 0,$ and $C_j = 2L_B\diam(\Om)|z_j - \mu_j|$ otherwise.  For $\epsilon > 0,$ following Theorem 2.7.11 of \ci{well:96}, the $\delta\epsilon$ bracketing numbers of these classes are all bounded by a multiple of $$(\epsilon\delta)^{-m}\sup_{\lVert x \rVert \leq B}N(c\delta\epsilon, B_\delta(m_\oplus(x)),d),$$ where $c$ and $m$ depend on the dimension $p$ and $\alpha$ only.  Then, letting $J$ be the integral in (U1), so that $J = O(\delta^{-\gamma})$ as $\delta \rightarrow 0$,
\begin{align*}
\tilde{J}(\delta) &= \int_0^1 \sqrt{1 + \log N_{[\,]}(\delta\epsilon, \mc{N}_\delta,\lVert \cdot \rVert_2)}\; d\epsilon = O\left(J + \int_0^1 \sqrt{-\log (\delta\epsilon)}\; d\epsilon\right) \\
&= O\left(\delta^{-\gamma} + \sqrt{-\log \delta}\right) = O(\delta^{-\gamma - \rho}) \quad \textrm{for any } \rho > 0,
\end{align*}
since $\sqrt{-\log \delta} = O(\delta^{-\rho})$ for any $\rho > 0$ as $\delta \rightarrow 0$.

Now, $H^0_\delta(z) = 2\diam(\Om)\delta$ and $H^j\delta(z) = 2\diam(\Om)\delta|z_j - \mu_j|$ are envelopes for $\mc{N}^j_\delta$, and $E\left((H^j_\delta(X))^2\right)$ are all $O(\delta^2)$.  Theorem  2.14.2 of \ci{well:96} provides the bound
\begin{align}
&E\left(\sup_{\lVert x \rVert_E \leq B} \sup_{d(\om, \mp(x)) < \delta} \left|\frac{1}{n}\sn\left(s_i(x)D_i(\om, x) - E(s_i(x)D_i(\om, x))\right)\right| \right) \label{eq: diff_bound4} \\
&\hspace{1.5cm} = O\left(\frac{\tilde{J}(\delta)\delta}{\sqrt{n}}\right)  = O\left(\frac{\delta^{1 - \gamma - \rho}}{\sqrt{n}}\right) \nonumber
\end{align}
for any $\rho > 0.$  By combining (\ref{eq: V_expand}), (\ref{eq: diff_bound3}) and (\ref{eq: diff_bound4}), for small $\delta$ and any $\rho > 0,$
\[
E\left(1_{A_R}\sup_{\lVert x \rVert_E \leq B}\sup_{d(\om, \mp(x)) < \delta} |V_n(\om) - V_n(\mp(x))|\right) \leq \frac{b\delta^{1-\gamma-\rho}}{\sqrt{n}}
\]
for some constant $b = b(\rho)$.

To finish, for any $\alpha' > \alpha,$ set $\rho = \alpha' - \alpha > 0$ and $q_n = n^{\frac{\alpha'}{4(\alpha' + \gamma - 1)}}$.  Following the pointwise rate argument, one can show that $$\sup_{\lVert x \rVert_E \leq B}d(\hmp(x),\mp(x))^{\alpha'/2} = O_p(q_n\inv),$$ so that
$$
\sup_{\lVert x \rVert_E \leq B}d(\hmp(x),\mp(x)) = O_p(q_n^{-2/\alpha'}) = O_p\left(n^{-\frac{1}{2(\alpha'+\gamma-1)}}\right).
$$

\end{proof}

\subsection{Proofs of results in Section~\ref{sec: local}}

For completeness, we include the elementary results of auxiliary Lemma~\ref{lma: muj} and its proof, which are well-known (Fan and Gijbels 1996).  The quantities of interest are $$\mu_j = E\left(K_h(X - x)(X-x)^j\right), \quad \tau_j(y) = E\left(K_h(X-x)(X-x)^j|Y=y\right),$$ with estimators $\hat{\mu}_j = \frac{1}{n}\sn K_h(X_i - x)(X_i-x)^j$, for $j = 0, 1, 2$.

\begin{Lemma}
  \label{lma: muj}
  Suppose (K0) and (L1) hold.  Then $$\mu_j = h^j\left[f(x)K_{1j} + hf'(x) K_{1(j+1)} + O(h^2)\right]$$ and $\hat{\mu}_j = \mu_j + O_p((h^{2j-1}n\inv)^{1/2})$ for $j = 0, 1, 2$.  Additionally, $$\tau_j(y) = h^j\left[g_y(x)K_{1j} + hg_y'(x)K_{1(j+1)} + O(h^2)\right],$$ where the $O(h^2)$ term is uniform over $y \in \Om$.
\end{Lemma}

\begin{proof}
  The statements regarding $\mu_j$  and $\tau_j(y)$ follow  from (K0) and (L1) using a second-order Taylor expansions of the densities $f$ and $g_y$.  Furthermore, $E(\hat{\mu}_j) = \mu_j$ is clear.  Next,
  \[
  E\left(K_h^2(X_i -x)(X_i-x)^{2j}\right) = h^{2j-1}\int K^2(u)u^{2j}f(x+hu)\;du = O(h^{2j-1}),
  \]
  so $\Var(\hat{\mu}_j) = O(h^{2j-1}n\inv)$, proving the result for the  $\hat{\mu}_j$.
\end{proof}

\begin{proof}[Proof of Theorem~\ref{thm: local_bias}]
First, we will show that ${\rm d}F_{Y|X}(x, y)/\dy = g_y(x)/f(x)$ for all $x$ such that $f(x) > 0$.  For any open set $U \subset \Om$, set
\[
a(x) = \int_U \frac{g_y(x)}{f(x)}\dy, \quad b(x) = \int_U {\rm d}F_{Y|X}(x, y).
\]
By assumption, both $a$ and $b$ are continuous.  Then for any $z \in \R$,
\begin{align*}
\int_{-\infty}^z a(x) f(x)\; dx &= \int_U \left(\int_{-\infty}^z g_y(x)\; dx\right)\dy \\
&= \int_U\left(\int_{-\infty}^z {\rm d}F_{X|Y}(x, y)\right)\dy = \int_{(-\infty,z)\times U}\;\dF \\
&= \int_{-\infty}^z \left(\int_U {\rm d}F_{Y|X}(x, y)\right) f(x)\; dx = \int_{-\infty}^z b(x)f(x)\; dx,
\end{align*}
proving the claim.

Next, using Lemma~\ref{lma: muj}
\begin{align*}
\int s(z, x, h){\rm d}F_{X|Y}(z|y) &= \frac{\mu_2\tau_0(y) - \mu_1\tau_1(y)}{\sigma_0^2} = \frac{g_y(x)}{f(x)} + O(h^2),
\end{align*}
where the error term is uniform over $y \in \Om$. Hence, using the previously established fact that ${\rm d}F_{Y|X}(x, y)/\dy = g_y(x)/f(x)$,
\begin{align*}
\tL(\om) &= \int d^2(y, \om)s(z, x, h)\dFz = \int d^2(y, \om)\frac{g_y(x)}{f(x)}\dy + O(h^2) \\
&= \int d^2(y, \om) {\rm d}F_{Y|X}(x, y) + O(h^2) = M_\oplus(\om, x) + O(h^2),
\end{align*}
where the error term is now uniform over $\om \in \Om$.  By (L0), we then have $d(\mp(x), \tlp(x)) = o(1)$ as $h = h_n \ra 0$.

Next, define $r_h = h^{-\frac{\beta_1}{\beta_1 - 1}}$ and set $S_{j, n} = \{\om: 2^{j-1} < r_hd(\om, \mp(x))^{\beta_1/2} \leq 2^j\}$. Let $I$ denote the indicator function.  Then,  for any $M > 0$, following similar arguments as the proof of Theorem~\ref{thm: rate} and using (L2), there exists $a > 0$ such that, for large $n$,
\begin{align*}
I\left(r_hd(\tlp(x), \mp(x))^{\beta_1/2} > 2^M\right) &\leq a\sum_{j \geq M} \frac{2^{2j(1 - \beta_1)/\beta_1}}{r_h^{2(1-\beta_1)/\beta_1}h^{-2}} \\
&\leq a\sum_{j \geq M}\left(\frac{1}{4^{(\beta_1-1)/\beta_1}}\right)^j ,
\end{align*}
which converges since $\beta_1 > 1$.  Thus, for some $M > 0$, we have $$d(\tlp(x), \mp(x)) \leq 2^{2M/\beta_1}h^{2/(\beta_1 - 1)}$$ for large $n$.
\end{proof}

\begin{Lemma}
  \label{lma: local_consistent}
  Suppose (K0), (L0) and (L1) hold, $\Om$ is bounded and that $h \ra 0$ and $nh \ra \infty$. Then $d(\tlp(x), \hlp(x)) = o_p(1)$.
\end{Lemma}

\begin{proof}
  We will show that $\tL - \hL \rightsquigarrow 0$ in $l^\infty(\Om)$.  Together with (L0), this will prove the result.

  To begin, write $s_i(x, h) = \sigma_0^{-2}K_h(X_i - x)\left[\mu_2 - \mu_1(X_i-x)\right]$.  Then the difference $\hL(\om) - \tL(\om)$ can be written as
  \begin{align}
  &\frac{1}{n}\sn\left[s_{in}(x, h) - s_i(x, h)\right]d^2(Y_i, \om)   \label{eq: L_expand}\\
  &\hspace{1.5cm} + \frac{1}{n}\sn\left(s_i(x, h)d^2(Y_i, \om) - E\left[s_i(x, h)d^2(Y_i, \om)\right]\right).\nonumber
  \end{align}
  Observe that $s_{in}(x, h) - s_i(x, h) = W_{0n}K_h(X_i - x) - W_{1n}K_h(X_i - x)(X_i-x)$, where
  \begin{equation}
  \label{eq: Ws2}
    W_{0n} = \frac{\hat{\mu}_2}{\hat{\sigma}_0^2} - \frac{\mu_2}{\sigma_0^2}, \quad W_{1n} = \frac{\hat{\mu}_1}{\hat{\sigma}_0^2} - \frac{\mu_1}{\sigma_0^2}
  \end{equation}
  Using the results of Lemma~\ref{lma: muj}, it follows that $W_{0n} = O_p((nh)^{-1/2})$ and $W_{1n} = O_p((nh^3)^{-1/2})$.  Since
  \begin{align*}
  E\left[K_h(X_i-x)(X_i-x)^jd^2(Y_i, \om)\right] & = O(h^j) \\
  E\left[K_h^2(X_i-x)(X_i-x)^{2j}d^4(Y_i,\om)\right] &= O(h^{2j-1}) 
  \end{align*}
  it follows that the first term in (\ref{eq: L_expand}) is $O_p((nh)^{-1/2})$.  One also finds that $E(s_i^2(x,h)) = O(h\inv)$, so that the second term in (\ref{eq: L_expand}) is also $O_p((nh)^{-1/2})$

  So far, we have shown that $\tL(\om) - \hL(\om) = o_p(1)$ for any $\om \in \Om$, since $nh \ra \infty$.  According to Theorem 1.5.4 in \ci{well:96}, the last thing we need to show is that, for any $\eta > 0$
  \[
  \limsup_n P\left(\sup_{d(\om_1, \om_2) < \delta}|(\tL - \hL)(\om_1) - (\tL - \hL)(\om_2)| > \eta\right) \ra 0 \quad \textrm{as} \; \delta \ra 0.
  \]
  Since $E(|s_i(x,h)|) = O(1)$ and $E(s_i^2(x,h)) = O(h\inv)$, $n\inv\sn |s_{in}(x,h)| = O_p(1)$.  Then $|\hL(\om_1) - \hL(\om_2)| \leq 2\diam(\Om)d(\om_1, \om_2)n\inv\sn|s_{in}(x, h)| = O_p(d(\om_1, \om_2))$.  Similarly, $|\tL(\om_1) - \tL(\om_2)| = O(d(\om_1, \om_2))$, which verifies the above.

\end{proof}

\begin{proof}[Proof of Theorem~\ref{thm: local_frechet}]
We adopt similar arguments as in the proof of Theorem~\ref{thm: rate}, with some adjustments.  Set $s_i(x,h) = K_h(X_i-x)\frac{\mu_2 - \mu_1(X_i-x)}{\sigma_0^2}$ and define $T_n(\om) = \hL(\om) - \tL(\om)$.  Letting $$D_i(\om,x) = d^2(Y_i, \om) - d^2(Y_i, \tlp(x)),$$
we have
\begin{align}
|T_n(\om) - T_n(\tlp(x))| &\leq \left|\frac{1}{n}\sn\left[s_{in}(x, h) - s_i(x, h)\right]D_i(\om,x)\right| \label{eq: T_expand}\\
&\hspace{1cm} + \left|\frac{1}{n}\sn\left(s_i(x,h)D_i(\om,x) - E\left[s_i(x, h)D_i(\om,x)\right]\right)\right|.\nonumber
\end{align}
Since $W_{0n}$ and $W_{1n}$ from (\ref{eq: Ws2}) are $O_p((nh)^{-1/2})$ and $O_p((nh^3)^{-1/2})$, respectively, and using the fact that $|D_i(\om,x)|\leq 2\diam(\Om)d(\om,\tlp(x))$,the first term on the right-hand side of (\ref{eq: T_expand}) is $O_p(d(\om, \tlp(x)))$, where the $O_p$ term is independent of $\om$ and $\tlp(x)$.  Thus, we can define $$B_R = \left\{\sup_{d(\om, \tlp(x)) < \delta}\left|\frac{1}{n}\sn\left[s_{in}(x, h) - s_i(x, h)\right]D_i(\om, x)\right| \leq R\delta(nh)^{-1/2}\right\}$$ for $R > 0$, so that $P(B_R^c) \ra 0$.

Next, to control the second term on the right-hand side of (\ref{eq: T_expand}), define the functions $g_\om: \R\times\Om \ra \R$ by
\[
g_\om(z, y) = \frac{1}{\sigma_0^2}K_h(z-x)[\mu_2 - \mu_1(z-x)]d^2(y, \om)
\]
and the corresponding function class
\[
\mathcal{M}_{n\delta} =\{g_\om - g_{\tlp(x)}: \; d(\om, \tlp(x)) < \delta\}.
\]
An envelope function for $\mathcal{M}_{n\delta}$ is
\[
G_{n\delta}(z) = \frac{2\diam(\Om)\delta}{\sigma_0^2}K_h(z-x)\left|\mu_2 - \mu_1(z-x) \right|,
\]
and $E(G_{n\delta}^2(X)) = O(\delta^2h\inv)$.  If $d(\mp(x), \tlp(x)) \leq \delta$, then $B_\delta(\tlp(x)) \subset B_{2\delta}(\mp(x))$, so that $N(\delta\epsilon, B_\delta(\tlp(x)), d) \leq N(\delta\epsilon/2, B_{2\delta}(\mp(x)), d)$ and by (P1),
\begin{eqnarray*}
&&\int_0^1 \sqrt{1 + \log N(\delta\epsilon, B_\delta(\tlp(x)), d)}\;d\epsilon \leq \\
&&\quad 4\int_0^1 \sqrt{1 + \log N(2\delta\epsilon, B_{2\delta}(\mp(x)), d)}\;d\epsilon = O(\delta^{-\gamma}).
\end{eqnarray*}
Using this fact together with Theorems 2.7.11 and 2.14.2 of \ci{well:96}, for small $\delta \geq d(\mp(x), \tlp(x))$,
\[
E\left(\sup_{d(\om, \tlp(x))<\delta}\left|\frac{1}{n}\sn\left(s_i(x,h)D_i(\om,x) - E\left[s_i(x,h)D_i(\om,x)\right]\right)\right|\right) = O\left(\delta^{1-\gamma}(nh)^{-1/2}\right).
\]
Combining this with (\ref{eq: T_expand}) and the definition of $B_R$,
\[
E\left(I_{B_R}\sup_{d(\om, \tlp(x)) < \delta} |T_n(\om) - T_n(\tlp(x))|\right) \leq \frac{a\delta^{1-\gamma}}{(nh)^{1/2}},
\]
where $I_{B_R}$ is the indicator function for the set $B_R$ and $a$ is a constant depending on $R$ and the entropy integral in (P1).

To finish, set $t_n = (nh)^{\frac{\beta_2}{4(\beta_2 + \gamma - 1)}}$ and define $$S_{j, n}(x) = \{\om: \; 2^{j-1} < t_nd(\om, \tlp(x))^{\beta_2/2} \leq 2^j\}.$$ Choose $\eta_2$ satisfying (L3) and such that (P1) is satisfied for any $\delta < \eta_2$.  By Theorem~\ref{thm: local_bias}, $d(\mp(x), \tlp(x)) = O(h^{2/(\beta_1-1)})$, and the bandwidth condition implies $h^{2/(\beta_1-1)} = O\big((nh)^{-1/(2(\beta_2+\gamma-1))}\big) = O(t_n^{-2/\beta_2})$.  Hence \newline
$d(\mp(x), \tlp(x)) \leq (2^M/t_n)^{2/\beta_2}$ for all $n$ once $M$ is sufficiently large, so the preceding bound applies with $\delta = (2^j/t_n)^{2/\beta_2}$ on each $S_{j,n}(x)$, $j \geq M$.  Set $\tilde{\eta} := (\eta_2/2)^{\beta_2/2}$.  For any integer $M$,
\begin{align}
&P\left(t_nd(\tlp(x), \hlp(x))^{\beta_2/2} > 2^M\right) \leq P(B_R^c) + P(2d(\tlp(x), \hlp(x)) > \eta) \label{eq: bound2}\\
& \hspace{1cm} + \sum_{\substack{j\geq M\\ 2^j \leq t_n\tilde{\eta}}}P\left(\left\{\sup_{\om \in S_{j,n}}|T_n(\om) - T_n(\tlp(x))| \geq C\frac{2^{2(j-1)}}{t_n^2}\right\}\cap B_R\right),\nonumber
\end{align}
where the second term goes to zero for any $\eta > 0$ by Lemma~\ref{lma: local_consistent}.  Since $$d(\om, \tlp(x)) < (2^j/t_n)^{2/\beta_2}$$ on $S_{j,n}(x)$, this implies that the sum on the right-hand side of (\ref{eq: bound2}) is bounded by
\[
4aC\inv\sum_{\substack{j \geq M \\ 2^j \leq t_n\tilde{\eta}}} \frac{2^{2j(1 - \beta_2 - \gamma)/\beta_2}}{t_n^{2(1-\beta_2-\gamma)/\beta_2}\sqrt{nh}} \leq 4aC\inv\sum_{j \geq M} \left(\frac{1}{4^{(\beta_2+\gamma-1)/\beta_2}}\right)^j,
\]
which converges since $\beta_2 + \gamma > 1$ (automatic, since $\beta_2 > 1$ and $\gamma \geq 0$).  Hence, $$d(\hlp(x), \tlp(x)) = O_p(t_n^{-2/\beta_2}) = O_p\left[ (nh)^{-\frac{1}{2(\beta_2 + \gamma - 1)}}\right].$$
\end{proof}

The proof of Corollary~\ref{cor:cor2} is straightforward and is omitted.

\subsection{Proofs of results in Section~\ref{sec: hilbert}}

\begin{proof}[Proof of Theorem~\ref{thm: hilbert1}]
Recall the notation introduced in Section~\ref{sec: hilbert}.  Observe that, when $\om$ ranges over $\Om$, the object $E\ip{Y}{\om}$ is a continuous linear operator under the assumption $E\lVert Y\rVert_\Om^2 < \infty$, so the existence and uniqueness of $\gamma_0$ follows by the Riesz representation theorem.  The same is true for the operator $E\ip{(X - \mu)Y}{\alpha}_p$, hence the existence and uniqueness of $\gamma_1$.  Next
\begin{align}
  E(s(X, x)\ip{Y}{\om}) &= E \ip{Y}{\om} + E\left[(X - \mu)^T\Si\inv(x - \mu)\ip{Y}{\om}\right]   \label{eq: ip_eq}
 \\
  &= \ip{\gamma_0}{\om} + E\ip{(X - \mu)Y}{(x - \mu)^T\Si\inv \om}_p \nonumber\\
  &= \ip{\beta_0}{\om} + \ip{(x - \mu)^T\Si\inv\gamma_1}{\om} \nonumber\\
  &= \ip{\beta_0 + \beta_1^T(x-\mu)}{\om}.\nonumber
\end{align}
Set $\tilde{\om} = \beta_0 + \beta_1^T(x-\mu)$ as in (\ref{eq: mp_hilsol}) and observe that $E (s(X, x)) = 1$.  Then by expanding the square, we have
\begin{align*}
M(\om, x) &= E\left(s(X,x)\lVert Y - \tilde{\om}\rVert_\Om^2 + 2s(X,x)\ip{Y - \tilde{\om}}{\tilde{\om} - \om}\right. \\
&\hspace{3cm} + \left.s(X,x)\lVert \tilde{\om} - \om\rVert_\Om^2\right) \\
&= M(\tilde{\om}, x) + 2\Big(E(s(X,x)\ip{Y}{\tilde{\om} - \om}) - \ip{\tilde{\om}}{\tilde{\om} - \om}\Big) + \lVert \tilde{\om} - \om \rVert_\Om^2.\\
\end{align*}
Hence, the middle term vanishes using (\ref{eq: ip_eq}) and we must have $\mp(x) = \tilde{\om}$.  As a weighted least squares problem, the empirical solution to (\ref{eq: mp_est}) is clearly \mbox{$\hmp(x) = n\inv \sn s_{in}(x)Y_i$}, which gives the proposed solution in (\ref{eq: mp_est_hilsol}).
\end{proof}

\begin{proof}[Proof of Theorem~\ref{thm: hilbert2}]
First, let $q = p+1$ and define $\beta = (\beta_0, \beta_1^T)^T$ and $\hbeta = (\hbeta_0, \hbeta_1^T)^T$.  By Theorem 1.8.4 in chapter 1.8 of \ci{well:96}, we only need to prove that, for all $\alpha \in \Om^q$, \mbox{$\sqrt{n}\ip{\hbeta - \beta}{\alpha}_q \rightsquigarrow \ip{\mathcal{G}}{\alpha}_q$} for the limiting process $\mathcal{G}$ and that $\sqrt{n}(\hbeta - \beta)$ is asymptotically finite dimensional.  The latter condition follows from the fact that $\bar{X} - \mu$ and $\lVert \Si\inv - \hSi\inv\rVert_F$ are $O_p(n^{-1/2})$ and by the assumptions on the moments of $\lVert Y\rVert_\Om$.  We will now prove the first condition.  This will require the definitions below, for any $m\times p$ matrix $A$ and symmetric $p\times p$ matrix $S$:
\begin{align*}
\vecop(A) = (A_{11},\ldots,A_{m1}, A_{12},\ldots, A_{m2},\ldots, A_{1p}, \ldots, A_{mp})^T, \\
\vechop(S) = (S_{11},\ldots,S_{p1},S_{22},\ldots,S_{p2},\ldots, S_{p,p-1}, S_{pp}).
\end{align*}

Let $\alpha \in \Om^q$ be fixed.  Define the $p\times p$ matrices $W_i = X_iX_i^T$ and $\eta_i$ with elements $(\eta_i)_{jk} = \ip{X_{ij}Y_i}{\alpha_{k+1}}$, and the vector $\xi_i \in \R^q$ with elements $\xi_{ij} = \ip{Y_i}{\alpha_j}$.  Also, define the vector $\rho \in \R^q$ with elements $\rho_j = \ip{\gamma_0}{\alpha_j}$ and the $p\times p$ matrix $\tau$ with elements $\tau_{jk} = \ip{\gamma_{1j}}{\alpha_{k+1}} + \mu_j\rho_{k+1}$.  Let
\begin{equation}
\label{eq: Z_vector}
Z_i = (X_i^T,\vechop(W_i)^T,\xi_i^T,\vecop(\eta_i)^T)^T.
\end{equation}
Then $Z_1,\ldots,Z_n$ are independently and identically distributed with expected value
\[
E(Z_i) = \left(\mu^T, \vechop(\Sigma +\mu\mu^T)^T, \rho^T, \vecop(\tau)^T\right)^T.
\]
Let $C_\alpha = \Var(Z_i)$.  Then by the central limit theorem,
\[
\sqrt{n}\left[\bar{Z} - E(Z_1)\right] \rightsquigarrow \mathcal{N}(0, C_\alpha).
\]

Next, for $a\in\R^p$, $c \in \R^q$, $G$ a symmetric $p\times p$ matrix and $H$ a $p\times p$ matrix, define the function
\[
g(a,\vechop(G),c,\vecop(H)) = c_1 + \sum_{j = 1}^p \sum_{k = 1}^p\left[\left(G - aa^T\right)\inv\right]_{jk}\left(H_{jk} - a_jc_{k+1}\right).
\]
Then
\begin{align*}
  g(E(Z_1)) &= \ip{\gamma_0}{\alpha_1} + \sum_{j = 1}^p\sum_{k = 1}^p (\Sigma\inv)_{jk}\ip{\gamma_{1j}}{\alpha_{k+1}} \\
  &= \ip{\beta_0}{\alpha_1} + \sum_{k = 1}^p\ip{\beta_{1k}}{\alpha_{k+1}} = \ip{\beta}{\alpha}_q
\end{align*}
and, similarly, $g(\bar{Z}) = \ip{\hat{\beta}}{\alpha}_q$.  Let $l_\alpha$ be the gradient vector of $g$ evaluated at $E(Z_1)$.  The elements of $l_\alpha$ can be computed as follows.  Let $\otimes$ denote the Kronecker product, $e_l\in \R^p$ be the vector of zeros with a single 1 in the $l$th entry, and $J^{lm}$ be the $p\times p$ matrix of zeros with a single 1 in the $(l, m)$th entry.  Set
\begin{align*}
\mathcal{A}^l &= \Sigma\inv(e_l^T\otimes \mu + \mu^T\otimes e_l)\Sigma\inv, \\
\mathcal{B}^{lm} &= -\Sigma\inv(J^{lm} + J^{ml} - J^{lm}J^{lm})\Sigma\inv.
\end{align*}
Let $s_l$ be the $l$th column of $\Sigma\inv$ and set $\alpha_{-1} = (\alpha_2,\ldots,\alpha_q)^T$.  The vector $l_\alpha$ can be formed using the values
\begin{align}
\frac{\partial g}{\partial a_l}(E(Z_1)) &= \ip{\mathcal{A}^l\gamma_1}{\alpha_{-1}}_p - \ip{\gamma_0}{\alpha_{-1}^Ts_l}, \quad 1\leq l \leq p,\label{eq: l_alpha}\\
\frac{\partial g}{\partial B_{lm}}(E(Z_1)) &= \ip{\mathcal{B}^{lm}\gamma_1}{\alpha_{-1}}_p, \quad 1\leq l \leq m \leq p,\nonumber\\
\frac{\partial g}{\partial c_1}(E(Z_1)) &= 1, \nonumber\\
\frac{\partial g}{\partial c_l}(E(Z_1)) &= -s_{l-1}^T\mu, \quad 2\leq l \leq q, \nonumber\\
\frac{\partial g}{\partial D_{lm}}(E(Z_1)) &= (\Sigma\inv)_{lm}, \quad 1 \leq l,m\leq p.\nonumber
\end{align}
Then the $\delta$-method yields
\[
\sqrt{n}\ip{\hbeta - \beta}{\alpha} \rightsquigarrow N(0, l_\alpha^TC_\alpha l_\alpha).
\]
\end{proof}

\begin{proof}[Proof of Corollary~\ref{cor: hilbert}]
Again, set $q = p+1$.  The first display in the corollary follows since $\sup_{x \in V_B} \lVert \hmp(x) - \mp(x)\rVert_\Om$ is bounded by
\[
\lVert \hbeta_0 - \beta_0\rVert_\Om + (\lVert\mu\rVert_E + B)\lVert \hbeta_1 - \beta_1\rVert_{\Om^p} + \lVert \xbar - \mu\rVert_E\lVert \hat{\beta}_1\rVert_\Om = O_p(n^{-1/2}).
\]

For the second result, note that Lemmas 1.5.2, 1.5.3 and Theorem 1.5.4 of \ci{well:96} can be generalized to the space $l_\Om^\infty(V_B)$.  Then we need to show that $\mathcal{M}_n$ is asymptotically tight and that, for any finite collection $x_1,\ldots,x_J \subset \R^p$, $(\mathcal{M}_n(x_1),\ldots,\mathcal{M}_n(x_J))$ converges weakly to the corresponding marginals of $\mathcal{M}$.

For simplicity, take $x_1,x_2 \in \R^p$.  Similar to the proof of Theorem~\ref{thm: hilbert1}, for fixed $\om \in \Om$, define $W_i = X_iX_i^T$, $\xi_i = \ip{Y_i}{\om}$ and $\eta_i \in \R^p$ with elements $\eta_{ij} = \ip{X_{ij}Y_i}{\om}$.  Also, define $\rho = \ip{\gamma_0}{\om}$, $\tau\in\R^p$ with elements $\tau_j = \ip{\gamma_{1j}}{\om} + \mu_j\rho$, and set $Z_i = (X_i^T, \vechop(W_i)^T, \xi_i, \eta_i^T)^T$.  Then $Z_1,\ldots,Z_n$ are independent with the same distribution and $E(Z_i) = (\mu^T, \vechop(\Sigma+\mu\mu^T)^T, \rho, \tau^T + \rho\mu^T)^T$.  Letting $C_\om = \Cov(Z_i)$, we have
\[
\sqrt{n}\left[\bar{Z} - E(Z_1)\right] \rightsquigarrow N(0, C_\om).
\]
For $a,c \in \R^p$, $b \in \R$ and $G$ a $p\times p$ symmetric matrix, define
\[
g_k(a, \vechop(G), b, c) = b + (x_k - a)^T(G - aa^T)\inv(c - ba),\quad k=1,2.
\]
It is easy to verify that $\hmp(x_k) = g_k(\bar{Z})$ and $\mp(x_k) = g_k(E(Z_1))$.  Define $r_{\om,k}$ to be the gradient of $g_k$ evaluated at $E(Z_1)$ and set $R_{\om} = (r_{\om,1},r_{\om,2})$.  Then the bivariate delta method gives
\[
(\mathcal{M}_n(x_1), \mathcal{M}_n(x_2))^T \rightsquigarrow N(0, R_\om^TC_\om R_\om).
\]
The process $\mathcal{M}$ is characterized by the distribution of its marginals, as given above.

For tightness, first let $\delta,\eps >0$ be given, define an orthonormal basis $\{e_j\}_{j =1}^\infty$ for $\Om$ and let $\Pi_J(\om) = \sum_{j = 1}^J \ip{\om}{e_j}e_j$ for any integer $J$ and $\om \in \Om$.  By combining Theorem~\ref{thm: hilbert1} and Lemma 1.8.1 of \ci{well:96}, there exists finite $J_0$ such that, with $\tilde{\mathcal{M}}_n(x) = \Pi_{J_0}(\mathcal{M}_n(x))$,
\[
\limsup_n P\left(\lVert \mathcal{M}_n - \tilde{\mathcal{M}}_n\rVert_{V_B}^2 > \delta\right)<\eps.
\]
Note that $\tilde{\mathcal{M}}_n(x) - \tilde{\mathcal{M}}_n(y) = \sum_{k=1}^p\Pi_{J_0}(\hat{\beta}_{1k} - \beta_{1k})(x_k-y_k)$ so that, for any $\eta > 0$,
\[
\lim_{\tau \ra 0}\limsup_n P\left(\sup_{\substack{\lVert x - y\rVert_E < \tau \\ x,y\in V_B}} \lVert \tilde{\mathcal{M}}_n(x) - \tilde{\mathcal{M}}_n(y)\rVert_\Om > \eta\right) \ra 0
\]
by again combining Theorem~\ref{thm: hilbert1} with Lemma 1.8.1 of \ci{well:96}.  This means that $\tilde{\mathcal{M}}_n$ is tight by Theorem 1.5.7 of \ci{well:96}, since $\tilde{\mathcal{M}}_n(x)$ takes values on the finite-dimensional Euclidean space spanned by the first $J_0$ basis functions $e_j\in\Om$.  For $A \subset l_\Om^\infty(V_B)$, define
\[
A^\delta = \{g \in l_\Om^\infty(V_B): \inf_{a \in A}\lVert a - g\rVert_{V_B} < \delta\}.
\]
Then there exists a compact set $K \subset l_\Om^\infty(V_B)$ such that $$\liminf_n P(\tilde{\mathcal{M}}_n \in K^\delta) \geq 1 - \eps$$ and, hence,
\begin{align*}
\liminf_n P(\mathcal{M}_n \in K^{2\delta}) &\geq \liminf_n P(\tilde{\mathcal{M}}_n \in K^\delta) \\
&\hspace{1cm} - \limsup_n P(\lVert \mathcal{M}_n - \tilde{\mathcal{M}}_n\rVert_{V_B} > \delta) \geq 1-2\eps,
\end{align*}
so $\mathcal{M}_n$ is asymptotically tight.
\end{proof}

\begin{proof}[Proof of Theorem~\ref{thm: hilbert_local}]
A.  For sufficiently small $h$, the conditions imply $\sigma^2_0 > 0$, so that $s(\cdot, x, h)$ is bounded since $K$ is bounded with compact support.  Since $E\lVert Y \rVert_\Om^2 < \infty$, both $E\left[\left|s(X, x, h)\right|\lVert Y\rVert_\Om\right]$ and $E\left[\left|s(X, x, h)\right|\lVert Y\rVert_\Om^2\right]$ are finite; in particular, $\tL$ is finite on $\Om$ and 
$\om \mapsto E\left[s(X, x, h)\ip{Y}{\om}\right]$ is a bounded linear functional and its Riesz Representer $\tlp(x)$ is well-defined and unique.  It follows that
\[
\tL(\om) = \tL(\tlp(x)) + \lVert \om - \tlp(x)\rVert_\Om^2,
\]
so that $\tlp(x)$ is the unique minimizer of $\tL$.

B.  Under the condition $nh \rightarrow \infty$, $\hat{\sigma}_0^2 > 0$ with probability approaching 1.  On this event, $\hL(\om) = \hL(\hlp(x)) + \lVert \om - \hlp(x)\rVert_\Om^2$, where $\hlp(x) = n\inv\sn s_{in}(x,h)Y_i$, establishing part B.
\end{proof}

\begin{proof}[Proof of Theorem~\ref{thm: hilbert_local2}]
Without loss of generality, assume that $K$ is supported on $[-1,1]$ and 
that $n$ is large enough so that $[x - h, x + h] \subset U$.  By
condition (i), $\mp(z) = E(Y \mid X = z)$ exists for $z \in U$; by
Theorem~\ref{thm: hilbert_local}, $\tlp(x) = E[s(X,x,h)Y]$ and, with probability tending to
one, $\hlp(x) = n\inv\sn s_{in}(x,h)Y_i$.  Let $s_i(x,h)$ and $W_{0n}$, $W_{1n}$ be as in the
proof of Lemma~\ref{lma: local_consistent}. 
Then $\hlp(x) - \tlp(x) = S_n + R_n,$ where $S_n = \frac{1}{n}\sn s_i(x,h)(Y_i - \tlp(x))$ 
and $R_n = W_{0n}A_n - W_{1n}B_n,$ with 
$
A_n = n\inv\sn K_h(X_i - x)(Y_i - \tlp(x)),$

$
B_n = n\inv\sn K_h(X_i - x)(X_i - x)(Y_i - \tlp(x)).
$  
Throughout, set
$\tau_h := \sup_{|z - x| \leq h}\lVert \mp(z) - \mp(x)\rVert_\Om$, so that $\tau_h \ra 0$ by
condition (ii).

We first show that $\lVert \tlp(x) - \mp(x)\rVert_\Om \ra 0$.  Conditioning on $X$,
$\tlp(x) - m_\oplus(x) = E[s(X,x,h)(\mp(X) - m_\oplus(x))]$, so that
$$
\lVert \tlp(x) - \mp(x)\rVert_\Om \leq E|s(X,x,h)|\,\tau_h = O(\tau_h) \ra 0,
$$
where the convergence follows from $E|s(X,x,h)| = O(1)$.  Consequently,
\begin{equation}
\label{eq:local_consistency}
\sup_{|z - x| \leq h}\lVert \mp(z) - \tlp(x)\rVert_\Om \leq 2\tau_h + o(1) \ra 0,
\end{equation} and
$E(\lVert Y - \tlp(x)\rVert_\Om^{2} \mid X = z)$ is bounded for $z \in U$ by condition (i).

Next, we show that $\sqrt{nh}\lVert R_n\rVert_\Om = o_p(1).$ By the stated conditions,
\[
\begin{split}
\lVert E(A_n)\rVert_\Om &\leq E[K_h(X - x)]\sup_{|z - x| \leq h}\lVert \mp(z) - \tlp(x)\rVert_\Om
= o(1), \\
E\lVert A_n - E(A_n)\rVert_\Om^2 &\leq n\inv E[K_h^2(X - x)\lVert Y - \tlp(x)\rVert_\Om^2]
= O((nh)\inv), 
\end{split}
\]
so that $\lVert A_n\rVert_\Om = o_p(1)$.  Similarly,
$\lVert E(B_n)\rVert_\Om = o(h)$ and $E\lVert B_n - E(B_n)\rVert_\Om^2 = O(h/n)$, so that
$\lVert B_n\rVert_\Om = o_p(h)$.  Hence
\[
\begin{split}
\sqrt{nh}\lVert R_n\rVert_\Om &\leq \sqrt{nh}|W_{0n}|\lVert A_n\rVert_\Om
+ \sqrt{nh}|W_{1n}|\lVert B_n\rVert_\Om \\
&= O_p(1)o_p(1) + O_p(h\inv)o_p(h) = o_p(1).
\end{split}
\]

To establish marginal convergence, fix $\om \in \Om$ and set
$$
Z_{ni} = s_i(x,h)\ip{Y_i - \tlp(x)}{\om},$$ 
so that
$\sqrt{nh}\ip{S_n}{\om} = \sqrt{h/n}\sn Z_{ni}$.  For each $n$, the $Z_{ni}$ are
independent and identically distributed with mean zero.  By
Lemma~\ref{lma: muj}, $\mu_2/\sigma_0^2 = f(x)\inv + o(1)$ and $\mu_1/\sigma_0^2 = O(1)$, so
that $s(z,x,h) = K_h(z - x)[f(x)\inv + o(1)]$ uniformly in $z$, and hence
$h\,E[s^2(X,x,h)] \ra R(K)/f(x)$.  Moreover, with 
$\Var(\ip{Y}{\om} \mid X = z) = \ip{\Gamma(z)\om}{\om}$, 
$$
\phi_n(z) := E(\ip{Y - \tlp(x)}{\om}^2 \mid X = z)
= \Var(\ip{Y}{\om} \mid X = z) + \ip{\mp(z) - \tlp(x)}{\om}^2
$$ satisfies
$\sup_{|z - x| \leq h}|\phi_n(z) - \ip{\Gamma(x)\om}{\om}| \ra 0$ by condition (iii) and
\eqref{eq:local_consistency}.  Therefore,
\[
h\Var(Z_{n1}) = h\,E[s^2(X,x,h)\phi_n(X)] \ra  \frac{R(K)}{f(x)}\Var(\ip{Y}{\om} \mid X = x).
\]
For the Lyapunov condition, condition (i) and 
\eqref{eq:local_consistency} and boundedness of $K$ imply that $E|Z_{n1}|^{2+\eps} = O(E|s(X,x,h)|^{2+\eps}) = O(h^{-(1+\eps)})$.
It follows that
\[
\sn E\left|\sqrt{h/n}\,Z_{ni}\right|^{2+\eps} = O\left(n(h/n)^{1+\eps/2}h^{-(1+\eps)}\right)
= O\left((nh)^{-\eps/2}\right) \ra 0.
\]
The Lindeberg--Feller central limit theorem, together with $\sqrt{nh}\lVert R_n\rVert_\Om = o_p(1)$ and Slutsky's theorem, then yield
$\sqrt{nh}\ip{\hlp(x) - \tlp(x)}{\om} \rightsquigarrow \ip{\mathcal{G}_x}{\om}$ for every
$\om \in \Om$.

It remains only to prove asymptotic finite dimensionality of the process.  Let $\{e_j\}_{j=1}^\infty$ be an orthonormal
basis of $\Om$ and $\Pi_J$ the orthogonal projection onto the span of $e_1,\ldots,e_J$, with
$\Pi_0 = 0$.  Since the summands of $S_n$ are i.i.d.\ with mean zero,
$E\lVert (I - \Pi_J)\sqrt{nh}S_n\rVert_\Om^2 = h\,E[s^2(X,x,h)\psi_{J,n}(X)]$, where
$$
\psi_{J,n}(z) := {\rm tr}\left((I - \Pi_J)\Gamma(z)(I - \Pi_J)\right)
+ \lVert (I - \Pi_J)(\mp(z) - \tlp(x))\rVert_\Om^2.
$$
Then by condition (iii) and \eqref{eq:local_consistency},
\[
\begin{split}
&\sup_{J \geq 0}\sup_{|z - x| \leq h}|\psi_{J,n}(z) - {\rm tr}((I - \Pi_J)\Gamma(x))| \\
&\hspace{.5cm}
\leq \sup_{|z - x| \leq h}\left\{\lVert \Gamma(z) - \Gamma(x)\rVert_{\rm tr}
+ \lVert \mp(z) - \tlp(x)\rVert_\Om^2\right\} \\
&\hspace{.5cm} \ra 0.
\end{split}
\]
Thus,
\[
\limsup_{n}E\lVert (I - \Pi_J)\sqrt{nh}S_n\rVert_\Om^2
= \frac{R(K)}{f(x)}{\rm tr}\left((I - \Pi_J)\Gamma(x)\right) \ra 0 \quad \textrm{as} \quad J \ra \infty,
\]
since $\Gamma(x)$ is trace class by condition (i).  By $\sqrt{nh}\lVert R_n\rVert_\Om = o_p(1)$ 
and Markov's inequality, we have that $\sqrt{nh}(\hlp(x) - \tlp(x))$ is asymptotically finite
dimensional. 
Combined with marginal convergence,
Theorem~1.8.4 of \ci{well:96} yields the result.
\end{proof}


\begin{thebibliography}{51}

\bibitem[\protect\citeauthoryear{Afsari}{2011}]{afsa:11}
\begin{barticle}[author]
\bauthor{\bsnm{Afsari},~\bfnm{Bijan}\binits{B.}}
(\byear{2011}).
\btitle{Riemannian Lp center of mass: Existence, uniqueness, and convexity}.
\bjournal{Proceedings of the American Mathematical Society}
\bvolume{139}
\bpages{655--673}.
\end{barticle}
\endbibitem

\bibitem[\protect\citeauthoryear{Allen et~al.}{2014}]{alle:14}
\begin{barticle}[author]
\bauthor{\bsnm{Allen},~\bfnm{Elena~A}\binits{E.~A.}},
  \bauthor{\bsnm{Damaraju},~\bfnm{Eswar}\binits{E.}},
  \bauthor{\bsnm{Plis},~\bfnm{Sergey~M}\binits{S.~M.}},
  \bauthor{\bsnm{Erhardt},~\bfnm{Erik~B}\binits{E.~B.}},
  \bauthor{\bsnm{Eichele},~\bfnm{Tom}\binits{T.}} \AND
  \bauthor{\bsnm{Calhoun},~\bfnm{Vince~D}\binits{V.~D.}}
(\byear{2014}).
\btitle{Tracking Whole-Brain Connectivity Dynamics in the Resting State}.
\bjournal{Cerebral Cortex}
\bvolume{24}
\bpages{663--676}.
\end{barticle}
\endbibitem

\bibitem[\protect\citeauthoryear{Arsigny et~al.}{2007}]{arsi:07}
\begin{barticle}[author]
\bauthor{\bsnm{Arsigny},~\bfnm{Vincent}\binits{V.}},
  \bauthor{\bsnm{Fillard},~\bfnm{Pierre}\binits{P.}},
  \bauthor{\bsnm{Pennec},~\bfnm{Xavier}\binits{X.}} \AND
  \bauthor{\bsnm{Ayache},~\bfnm{Nicholas}\binits{N.}}
(\byear{2007}).
\btitle{Geometric means in a novel vector space structure on symmetric
  positive-definite matrices}.
\bjournal{SIAM Journal on Matrix Analysis and Applications}
\bvolume{29}
\bpages{328--347}.
\end{barticle}
\endbibitem

\bibitem[\protect\citeauthoryear{Barden, Le and Owen}{2013}]{bard:13}
\begin{barticle}[author]
\bauthor{\bsnm{Barden},~\bfnm{Dennis}\binits{D.}},
  \bauthor{\bsnm{Le},~\bfnm{Huiling}\binits{H.}} \AND
  \bauthor{\bsnm{Owen},~\bfnm{Megan}\binits{M.}}
(\byear{2013}).
\btitle{Central limit theorems for {F}r{\'e}chet means in the space of
  phylogenetic trees}.
\bjournal{Electronic Journal of Probability}
\bvolume{18}
\bpages{1--25}.
\end{barticle}
\endbibitem

\bibitem[\protect\citeauthoryear{Bhattacharya and Patrangenaru}{2003}]{bhat:03}
\begin{barticle}[author]
\bauthor{\bsnm{Bhattacharya},~\bfnm{R.}\binits{R.}} \AND
  \bauthor{\bsnm{Patrangenaru},~\bfnm{V.}\binits{V.}}
(\byear{2003}).
\btitle{Large sample theory of intrinsic and extrinsic sample means on
  manifolds - {I}}.
\bjournal{Annals of Statistics}
\bvolume{31}
\bpages{1--29}.
\end{barticle}
\endbibitem

\bibitem[\protect\citeauthoryear{Bhattacharya et~al.}{2012}]{bhat:12}
\begin{barticle}[author]
\bauthor{\bsnm{Bhattacharya},~\bfnm{Rabindra~N}\binits{R.~N.}},
  \bauthor{\bsnm{Ellingson},~\bfnm{L}\binits{L.}},
  \bauthor{\bsnm{Liu},~\bfnm{X}\binits{X.}},
  \bauthor{\bsnm{Patrangenaru},~\bfnm{V}\binits{V.}} \AND
  \bauthor{\bsnm{Crane},~\bfnm{M}\binits{M.}}
(\byear{2012}).
\btitle{Extrinsic analysis on manifolds is computationally faster than
  intrinsic analysis with applications to quality control by machine vision}.
\bjournal{Applied Stochastic Models in Business and Industry}
\bvolume{28}
\bpages{222--235}.
\end{barticle}
\endbibitem

\bibitem[\protect\citeauthoryear{Borsdorf and Higham}{2010}]{bors:10}
\begin{barticle}[author]
\bauthor{\bsnm{Borsdorf},~\bfnm{R{\"u}diger}\binits{R.}} \AND
  \bauthor{\bsnm{Higham},~\bfnm{Nicholas~J}\binits{N.~J.}}
(\byear{2010}).
\btitle{A preconditioned {N}ewton algorithm for the nearest correlation
  matrix}.
\bjournal{IMA Journal of Numerical Analysis}
\bvolume{30}
\bpages{94--107}.
\end{barticle}
\endbibitem

\bibitem[\protect\citeauthoryear{Boumal et~al.}{2014}]{boum:14}
\begin{barticle}[author]
\bauthor{\bsnm{Boumal},~\bfnm{Nicolas}\binits{N.}},
  \bauthor{\bsnm{Mishra},~\bfnm{Bamdev}\binits{B.}},
  \bauthor{\bsnm{Absil},~\bfnm{Pierre-Antoine}\binits{P.-A.}},
  \bauthor{\bsnm{Sepulchre},~\bfnm{Rodolphe}\binits{R.}} \betal{et~al.}
(\byear{2014}).
\btitle{Manopt, a matlab toolbox for optimization on manifolds}.
\bjournal{Journal of Machine Learning Research}
\bvolume{15}
\bpages{1455--1459}.
\end{barticle}
\endbibitem

\bibitem[\protect\citeauthoryear{Bradley}{1968}]{brad:68}
\begin{bbook}[author]
\bauthor{\bsnm{Bradley},~\bfnm{James~V}\binits{J.~V.}}
(\byear{1968}).
\btitle{Distribution-free {S}tatistical {T}ests}.
\bpublisher{NJ, Prentice-Hall}.
\end{bbook}
\endbibitem

\bibitem[\protect\citeauthoryear{Chang}{1989}]{chan:89}
\begin{barticle}[author]
\bauthor{\bsnm{Chang},~\bfnm{Ted}\binits{T.}}
(\byear{1989}).
\btitle{Spherical regression with errors in variables}.
\bjournal{Annals of Statistics}
\bvolume{17}
\bpages{293--306}.
\end{barticle}
\endbibitem

\bibitem[\protect\citeauthoryear{Cornea et~al.}{2016}]{corn:17}
\begin{barticle}[author]
\bauthor{\bsnm{Cornea},~\bfnm{Emil}\binits{E.}},
  \bauthor{\bsnm{Zhu},~\bfnm{Hongtu}\binits{H.}},
  \bauthor{\bsnm{Kim},~\bfnm{Peter}\binits{P.}} \AND
  \bauthor{\bsnm{Ibrahim},~\bfnm{Joseph~G}\binits{J.~G.}}
(\byear{2016}).
\btitle{Regression models on Riemannian symmetric spaces}.
\bjournal{Journal of the Royal Statistical Society: Series B}.
\end{barticle}
\endbibitem

\bibitem[\protect\citeauthoryear{Craven and Wahba}{1979}]{crav:79}
\begin{barticle}[author]
\bauthor{\bsnm{Craven},~\bfnm{Peter}\binits{P.}} \AND
  \bauthor{\bsnm{Wahba},~\bfnm{Grace}\binits{G.}}
(\byear{1979}).
\btitle{Smoothing noisy data with spline functions.}
\bjournal{Numerical Mathematics}
\bvolume{31}
\bpages{377--403}.
\bmrnumber{MR516581 (81g:65018)}
\end{barticle}
\endbibitem

\bibitem[\protect\citeauthoryear{Davis et~al.}{2007}]{davi:07}
\begin{binproceedings}[author]
\bauthor{\bsnm{Davis},~\bfnm{Bradley~C}\binits{B.~C.}},
  \bauthor{\bsnm{Fletcher},~\bfnm{P~Thomas}\binits{P.~T.}},
  \bauthor{\bsnm{Bullitt},~\bfnm{Elizabeth}\binits{E.}} \AND
  \bauthor{\bsnm{Joshi},~\bfnm{S}\binits{S.}}
(\byear{2007}).
\btitle{Population shape regression from random design data}.
In \bbooktitle{ICCV 2007. IEEE 11th International Conference on Computer
  Vision}
\bpages{1--7}.
\end{binproceedings}
\endbibitem

\bibitem[\protect\citeauthoryear{Fan and Gijbels}{1996}]{fan:96}
\begin{bbook}[author]
\bauthor{\bsnm{Fan},~\bfnm{J.}\binits{J.}} \AND
  \bauthor{\bsnm{Gijbels},~\bfnm{I.}\binits{I.}}
(\byear{1996}).
\btitle{Local Polynomial Modelling and its Applications}.
\bpublisher{Chapman \& Hall}, \baddress{London}.
\bmrnumber{MR1383587 (97f:62063)}
\end{bbook}
\endbibitem

\bibitem[\protect\citeauthoryear{Faraway}{1997}]{fara:97}
\begin{barticle}[author]
\bauthor{\bsnm{Faraway},~\bfnm{Julian~J.}\binits{J.~J.}}
(\byear{1997}).
\btitle{Regression analysis for a functional response}.
\bjournal{Technometrics}
\bvolume{39}
\bpages{254--261}.
\bmrnumber{MR1462586}
\end{barticle}
\endbibitem

\bibitem[\protect\citeauthoryear{Faraway}{2014}]{fara:14}
\begin{barticle}[author]
\bauthor{\bsnm{Faraway},~\bfnm{Julian~J}\binits{J.~J.}}
(\byear{2014}).
\btitle{Regression for non-{E}uclidean data using distance matrices}.
\bjournal{Journal of Applied Statistics}
\bvolume{41}
\bpages{2342--2357}.
\end{barticle}
\endbibitem

\bibitem[\protect\citeauthoryear{Ferreira and Busatto}{2013}]{ferr:13}
\begin{barticle}[author]
\bauthor{\bsnm{Ferreira},~\bfnm{Luiz~Kobuti}\binits{L.~K.}} \AND
  \bauthor{\bsnm{Busatto},~\bfnm{Geraldo~F}\binits{G.~F.}}
(\byear{2013}).
\btitle{Resting-state functional connectivity in normal brain aging}.
\bjournal{Neuroscience \& Biobehavioral Reviews}
\bvolume{37}
\bpages{384--400}.
\end{barticle}
\endbibitem

\bibitem[\protect\citeauthoryear{Ferreira et~al.}{2013}]{ferr:13:2}
\begin{barticle}[author]
\bauthor{\bsnm{Ferreira},~\bfnm{Ricardo}\binits{R.}},
  \bauthor{\bsnm{Xavier},~\bfnm{Jo{\~a}o}\binits{J.}},
  \bauthor{\bsnm{Costeira},~\bfnm{Jo{\~a}o~P}\binits{J.~P.}} \AND
  \bauthor{\bsnm{Barroso},~\bfnm{Victor}\binits{V.}}
(\byear{2013}).
\btitle{Newton algorithms for Riemannian distance related problems on connected
  locally symmetric manifolds}.
\bjournal{IEEE Journal of Selected Topics in Signal Processing}
\bvolume{7}
\bpages{634--645}.
\end{barticle}
\endbibitem

\bibitem[\protect\citeauthoryear{Fisher}{1995}]{fish:95}
\begin{bbook}[author]
\bauthor{\bsnm{Fisher},~\bfnm{Nicholas~I}\binits{N.~I.}}
(\byear{1995}).
\btitle{Statistical analysis of circular data}.
\bpublisher{Cambridge University Press}.
\end{bbook}
\endbibitem

\bibitem[\protect\citeauthoryear{Fisher, Lewis and Embleton}{1987}]{fish:87}
\begin{bbook}[author]
\bauthor{\bsnm{Fisher},~\bfnm{Nicholas~I}\binits{N.~I.}},
  \bauthor{\bsnm{Lewis},~\bfnm{Toby}\binits{T.}} \AND
  \bauthor{\bsnm{Embleton},~\bfnm{Brian~JJ}\binits{B.~J.}}
(\byear{1987}).
\btitle{Statistical analysis of spherical data}.
\bpublisher{Cambridge University Press}.
\end{bbook}
\endbibitem

\bibitem[\protect\citeauthoryear{Fletcher}{2013}]{flet:13}
\begin{barticle}[author]
\bauthor{\bsnm{Fletcher},~\bfnm{P~Thomas}\binits{P.~T.}}
(\byear{2013}).
\btitle{Geodesic regression and the theory of least squares on {R}iemannian
  manifolds}.
\bjournal{International Journal of Computer Vision}
\bvolume{105}
\bpages{171--185}.
\end{barticle}
\endbibitem

\bibitem[\protect\citeauthoryear{Fr{\'e}chet}{1948}]{frec:48}
\begin{binproceedings}[author]
\bauthor{\bsnm{Fr{\'e}chet},~\bfnm{Maurice}\binits{M.}}
(\byear{1948}).
\btitle{Les {\'e}l{\'e}ments al{\'e}atoires de nature quelconque dans un espace
  distanci{\'e}}.
In \bbooktitle{Annales de l'Institut Henri Poincar{\'e}}
\bvolume{10}
\bpages{215--310}.
\end{binproceedings}
\endbibitem

\bibitem[\protect\citeauthoryear{Hein}{2009}]{hein:09}
\begin{binproceedings}[author]
\bauthor{\bsnm{Hein},~\bfnm{Matthias}\binits{M.}}
(\byear{2009}).
\btitle{Robust Nonparametric Regression with Metric-Space valued Output}.
In \bbooktitle{Advances in Neural Information Processing Systems}
\bpages{718--726}.
\end{binproceedings}
\endbibitem

\bibitem[\protect\citeauthoryear{Higgins}{2004}]{higg:04}
\begin{bbook}[author]
\bauthor{\bsnm{Higgins},~\bfnm{James~J}\binits{J.~J.}}
(\byear{2004}).
\btitle{An introduction to modern nonparametric statistics}.
\bpublisher{Brooks/Cole Pacific Grove, CA}.
\end{bbook}
\endbibitem

\bibitem[\protect\citeauthoryear{Higham}{2002}]{high:02}
\begin{barticle}[author]
\bauthor{\bsnm{Higham},~\bfnm{Nicholas~J}\binits{N.~J.}}
(\byear{2002}).
\btitle{Computing the nearest correlation matrix -- a problem from finance}.
\bjournal{IMA Journal of Numerical Analysis}
\bvolume{22}
\bpages{329--343}.
\end{barticle}
\endbibitem

\bibitem[\protect\citeauthoryear{Hinkle et~al.}{2012}]{hink:12}
\begin{bincollection}[author]
\bauthor{\bsnm{Hinkle},~\bfnm{Jacob}\binits{J.}},
  \bauthor{\bsnm{Muralidharan},~\bfnm{Prasanna}\binits{P.}},
  \bauthor{\bsnm{Fletcher},~\bfnm{P~Thomas}\binits{P.~T.}} \AND
  \bauthor{\bsnm{Joshi},~\bfnm{Sarang}\binits{S.}}
(\byear{2012}).
\btitle{Polynomial regression on {R}iemannian manifolds}.
In \bbooktitle{Computer Vision--ECCV 2012}
\bpages{1--14}.
\bpublisher{Springer}.
\end{bincollection}
\endbibitem

\bibitem[\protect\citeauthoryear{Le and Barden}{2014}]{le:14}
\begin{barticle}[author]
\bauthor{\bsnm{Le},~\bfnm{H}\binits{H.}} \AND
  \bauthor{\bsnm{Barden},~\bfnm{Dennis}\binits{D.}}
(\byear{2014}).
\btitle{On the measure of the cut locus of a Fr{\'e}chet mean}.
\bjournal{Bulletin of the London Mathematical Society}
\bpages{bdu025}.
\end{barticle}
\endbibitem

\bibitem[\protect\citeauthoryear{Lee, Smyser and Shimony}{2013}]{lee:13}
\begin{barticle}[author]
\bauthor{\bsnm{Lee},~\bfnm{MH}\binits{M.}},
  \bauthor{\bsnm{Smyser},~\bfnm{CD}\binits{C.}} \AND
  \bauthor{\bsnm{Shimony},~\bfnm{JS}\binits{J.}}
(\byear{2013}).
\btitle{Resting-state f{MRI}: a review of methods and clinical applications}.
\bjournal{American Journal of Neuroradiology}
\bvolume{34}
\bpages{1866--1872}.
\end{barticle}
\endbibitem

\bibitem[\protect\citeauthoryear{Lehmann and D'Abrera}{2006}]{lehm:06}
\begin{bbook}[author]
\bauthor{\bsnm{Lehmann},~\bfnm{Erich~Leo}\binits{E.~L.}} \AND
  \bauthor{\bsnm{D'Abrera},~\bfnm{Howard~JM}\binits{H.~J.}}
(\byear{2006}).
\btitle{Nonparametrics: statistical methods based on ranks}.
\bpublisher{Springer New York}.
\end{bbook}
\endbibitem

\bibitem[\protect\citeauthoryear{Lin et~al.}{2015}]{lin:15}
\begin{barticle}[author]
\bauthor{\bsnm{Lin},~\bfnm{Lizhen}\binits{L.}},
  \bauthor{\bsnm{Thomas},~\bfnm{Brian~St}\binits{B.~S.}},
  \bauthor{\bsnm{Zhu},~\bfnm{Hongtu}\binits{H.}} \AND
  \bauthor{\bsnm{Dunson},~\bfnm{David~B}\binits{D.~B.}}
(\byear{2015}).
\btitle{Extrinsic local regression on manifold-valued data}.
\bjournal{arXiv preprint arXiv:1508.02201}.
\end{barticle}
\endbibitem

\bibitem[\protect\citeauthoryear{Marron and Alonso}{2014}]{marr:14}
\begin{barticle}[author]
\bauthor{\bsnm{Marron},~\bfnm{J~Steve}\binits{J.~S.}} \AND
  \bauthor{\bsnm{Alonso},~\bfnm{Andr{\'e}s~M}\binits{A.~M.}}
(\byear{2014}).
\btitle{Overview of object oriented data analysis}.
\bjournal{Biometrical Journal}
\bvolume{56}
\bpages{732--753}.
\end{barticle}
\endbibitem

\bibitem[\protect\citeauthoryear{Marx and Eilers}{1996}]{marx:96}
\begin{barticle}[author]
\bauthor{\bsnm{Marx},~\bfnm{B.}\binits{B.}} \AND
  \bauthor{\bsnm{Eilers},~\bfnm{B.}\binits{B.}}
(\byear{1996}).
\btitle{Flexible smoothing with {B}-splines and penalties (with comments and
  rejoinder)}.
\bjournal{Statistical Science}
\bvolume{11}
\bpages{89--121}.
\end{barticle}
\endbibitem

\bibitem[\protect\citeauthoryear{Mevel et~al.}{2013}]{meve:13}
\begin{barticle}[author]
\bauthor{\bsnm{Mevel},~\bfnm{Katell}\binits{K.}},
  \bauthor{\bsnm{Landeau},~\bfnm{Brigitte}\binits{B.}},
  \bauthor{\bsnm{Fouquet},~\bfnm{Marine}\binits{M.}},
  \bauthor{\bsnm{La~Joie},~\bfnm{Renaud}\binits{R.}},
  \bauthor{\bsnm{Villain},~\bfnm{Nicolas}\binits{N.}},
  \bauthor{\bsnm{M{\'e}zenge},~\bfnm{Florence}\binits{F.}},
  \bauthor{\bsnm{Perrotin},~\bfnm{Audrey}\binits{A.}},
  \bauthor{\bsnm{Eustache},~\bfnm{Francis}\binits{F.}},
  \bauthor{\bsnm{Desgranges},~\bfnm{Beatrice}\binits{B.}} \AND
  \bauthor{\bsnm{Ch{\'e}telat},~\bfnm{Ga{\"e}l}\binits{G.}}
(\byear{2013}).
\btitle{Age effect on the default mode network, inner thoughts, and cognitive
  abilities}.
\bjournal{Neurobiology of Aging}
\bvolume{34}
\bpages{1292--1301}.
\end{barticle}
\endbibitem

\bibitem[\protect\citeauthoryear{Niethammer, Huang and Vialard}{2011}]{niet:11}
\begin{bincollection}[author]
\bauthor{\bsnm{Niethammer},~\bfnm{Marc}\binits{M.}},
  \bauthor{\bsnm{Huang},~\bfnm{Yang}\binits{Y.}} \AND
  \bauthor{\bsnm{Vialard},~\bfnm{Fran{\c{c}}ois-Xavier}\binits{F.-X.}}
(\byear{2011}).
\btitle{Geodesic regression for image time-series}.
In \bbooktitle{Medical Image Computing and Computer-Assisted
  Intervention--MICCAI 2011}
\bpages{655--662}.
\bpublisher{Springer}.
\end{bincollection}
\endbibitem

\bibitem[\protect\citeauthoryear{Onoda, Ishihara and Yamaguchi}{2012}]{onod:12}
\begin{barticle}[author]
\bauthor{\bsnm{Onoda},~\bfnm{Keiichi}\binits{K.}},
  \bauthor{\bsnm{Ishihara},~\bfnm{Masaki}\binits{M.}} \AND
  \bauthor{\bsnm{Yamaguchi},~\bfnm{Shuhei}\binits{S.}}
(\byear{2012}).
\btitle{Decreased functional connectivity by aging is associated with cognitive
  decline}.
\bjournal{Journal of Cognitive Neuroscience}
\bvolume{24}
\bpages{2186--2198}.
\end{barticle}
\endbibitem

\bibitem[\protect\citeauthoryear{Panaretos and Zemel}{2016}]{pana:16}
\begin{barticle}[author]
\bauthor{\bsnm{Panaretos},~\bfnm{Victor~M}\binits{V.~M.}} \AND
  \bauthor{\bsnm{Zemel},~\bfnm{Yoav}\binits{Y.}}
(\byear{2016}).
\btitle{Amplitude and phase variation of point processes}.
\bjournal{The Annals of Statistics}
\bvolume{44}
\bpages{771--812}.
\end{barticle}
\endbibitem

\bibitem[\protect\citeauthoryear{Patrangenaru and Ellingson}{2015}]{patr:15}
\begin{bbook}[author]
\bauthor{\bsnm{Patrangenaru},~\bfnm{Victor}\binits{V.}} \AND
  \bauthor{\bsnm{Ellingson},~\bfnm{Leif}\binits{L.}}
(\byear{2015}).
\btitle{Nonparametric Statistics on Manifolds and Their Applications to Object
  Data Analysis}.
\bpublisher{CRC Press}.
\end{bbook}
\endbibitem

\bibitem[\protect\citeauthoryear{Pelletier}{2006}]{pell:06}
\begin{barticle}[author]
\bauthor{\bsnm{Pelletier},~\bfnm{Bruno}\binits{B.}}
(\byear{2006}).
\btitle{Non-parametric regression estimation on closed {R}iemannian manifolds}.
\bjournal{Journal of Nonparametric Statistics}
\bvolume{18}
\bpages{57--67}.
\end{barticle}
\endbibitem

\bibitem[\protect\citeauthoryear{Pigoli et~al.}{2014}]{pigo:14}
\begin{barticle}[author]
\bauthor{\bsnm{Pigoli},~\bfnm{Davide}\binits{D.}},
  \bauthor{\bsnm{Aston},~\bfnm{John~AD}\binits{J.~A.}},
  \bauthor{\bsnm{Dryden},~\bfnm{Ian~L}\binits{I.~L.}} \AND
  \bauthor{\bsnm{Secchi},~\bfnm{Piercesare}\binits{P.}}
(\byear{2014}).
\btitle{Distances and inference for covariance operators}.
\bjournal{Biometrika}
\bvolume{101}
\bpages{409--422}.
\end{barticle}
\endbibitem

\bibitem[\protect\citeauthoryear{Prentice}{1989}]{pren:89}
\begin{barticle}[author]
\bauthor{\bsnm{Prentice},~\bfnm{Michael~J}\binits{M.~J.}}
(\byear{1989}).
\btitle{Spherical regression on matched pairs of orientation statistics}.
\bjournal{Journal of the Royal Statistical Society: Series B}
\bpages{241--248}.
\end{barticle}
\endbibitem

\bibitem[\protect\citeauthoryear{Qi and Sun}{2006}]{qi:06}
\begin{barticle}[author]
\bauthor{\bsnm{Qi},~\bfnm{Houduo}\binits{H.}} \AND
  \bauthor{\bsnm{Sun},~\bfnm{Defeng}\binits{D.}}
(\byear{2006}).
\btitle{A quadratically convergent {N}ewton method for computing the nearest
  correlation matrix}.
\bjournal{SIAM Journal on Matrix Analysis and Applications}
\bvolume{28}
\bpages{360--385}.
\end{barticle}
\endbibitem

\bibitem[\protect\citeauthoryear{Sheline and Raichle}{2013}]{shel:13}
\begin{barticle}[author]
\bauthor{\bsnm{Sheline},~\bfnm{Yvette~I}\binits{Y.~I.}} \AND
  \bauthor{\bsnm{Raichle},~\bfnm{Marcus~E}\binits{M.~E.}}
(\byear{2013}).
\btitle{Resting state functional connectivity in preclinical {A}lzheimer's
  disease}.
\bjournal{Biological Psychiatry}
\bvolume{74}
\bpages{340--347}.
\end{barticle}
\endbibitem

\bibitem[\protect\citeauthoryear{Shi et~al.}{2009}]{shi:09}
\begin{bincollection}[author]
\bauthor{\bsnm{Shi},~\bfnm{Xiaoyan}\binits{X.}},
  \bauthor{\bsnm{Styner},~\bfnm{Martin}\binits{M.}},
  \bauthor{\bsnm{Lieberman},~\bfnm{Jeffrey}\binits{J.}},
  \bauthor{\bsnm{Ibrahim},~\bfnm{Joseph~G}\binits{J.~G.}},
  \bauthor{\bsnm{Lin},~\bfnm{Weili}\binits{W.}} \AND
  \bauthor{\bsnm{Zhu},~\bfnm{Hongtu}\binits{H.}}
(\byear{2009}).
\btitle{Intrinsic regression models for manifold-valued data}.
In \bbooktitle{Medical Image Computing and Computer-Assisted
  Intervention--MICCAI 2009}
\bpages{192--199}.
\bpublisher{Springer}.
\end{bincollection}
\endbibitem

\bibitem[\protect\citeauthoryear{Steinke and Hein}{2009}]{stei:09}
\begin{binproceedings}[author]
\bauthor{\bsnm{Steinke},~\bfnm{Florian}\binits{F.}} \AND
  \bauthor{\bsnm{Hein},~\bfnm{Matthias}\binits{M.}}
(\byear{2009}).
\btitle{Non-parametric regression between manifolds}.
In \bbooktitle{Advances in Neural Information Processing Systems}
\bpages{1561--1568}.
\end{binproceedings}
\endbibitem

\bibitem[\protect\citeauthoryear{Steinke, Hein and
  Sch{\"o}lkopf}{2010}]{stei:10}
\begin{barticle}[author]
\bauthor{\bsnm{Steinke},~\bfnm{Florian}\binits{F.}},
  \bauthor{\bsnm{Hein},~\bfnm{Matthias}\binits{M.}} \AND
  \bauthor{\bsnm{Sch{\"o}lkopf},~\bfnm{Bernhard}\binits{B.}}
(\byear{2010}).
\btitle{Nonparametric regression between general {R}iemannian manifolds}.
\bjournal{SIAM Journal on Imaging Sciences}
\bvolume{3}
\bpages{527--563}.
\end{barticle}
\endbibitem

\bibitem[\protect\citeauthoryear{Su et~al.}{2012}]{su:12}
\begin{barticle}[author]
\bauthor{\bsnm{Su},~\bfnm{Jingyong}\binits{J.}},
  \bauthor{\bsnm{Dryden},~\bfnm{Ian~L}\binits{I.~L.}},
  \bauthor{\bsnm{Klassen},~\bfnm{Eric}\binits{E.}},
  \bauthor{\bsnm{Le},~\bfnm{Huiling}\binits{H.}} \AND
  \bauthor{\bsnm{Srivastava},~\bfnm{Anuj}\binits{A.}}
(\byear{2012}).
\btitle{Fitting smoothing splines to time-indexed, noisy points on nonlinear
  manifolds}.
\bjournal{Image and Vision Computing}
\bvolume{30}
\bpages{428--442}.
\end{barticle}
\endbibitem

\bibitem[\protect\citeauthoryear{Takatsu}{2011}]{taka:11}
\begin{barticle}[author]
\bauthor{\bsnm{Takatsu},~\bfnm{Asuka}\binits{A.}}
(\byear{2011}).
\btitle{Wasserstein geometry of {G}aussian measures}.
\bjournal{Osaka Journal of Mathematics}
\bvolume{48}
\bpages{1005--1026}.
\end{barticle}
\endbibitem

\bibitem[\protect\citeauthoryear{van~der Vaart and Wellner}{1996}]{well:96}
\begin{bbook}[author]
\bauthor{\bparticle{van~der} \bsnm{Vaart},~\bfnm{Aad}\binits{A.}} \AND
  \bauthor{\bsnm{Wellner},~\bfnm{John}\binits{J.}}
(\byear{1996}).
\btitle{Weak Convergence and Empirical Processes}.
\bpublisher{Springer, New York}.
\end{bbook}
\endbibitem

\bibitem[\protect\citeauthoryear{Wang et~al.}{2007}]{wang:07:1}
\begin{barticle}[author]
\bauthor{\bsnm{Wang},~\bfnm{Haonan}\binits{H.}},
  \bauthor{\bsnm{Marron},~\bfnm{JS}\binits{J.}} \betal{et~al.}
(\byear{2007}).
\btitle{Object oriented data analysis: Sets of trees}.
\bjournal{Annals of Statistics}
\bvolume{35}
\bpages{1849--1873}.
\end{barticle}
\endbibitem

\bibitem[\protect\citeauthoryear{Yuan et~al.}{2012}]{yuan:12}
\begin{barticle}[author]
\bauthor{\bsnm{Yuan},~\bfnm{Ying}\binits{Y.}},
  \bauthor{\bsnm{Zhu},~\bfnm{Hongtu}\binits{H.}},
  \bauthor{\bsnm{Lin},~\bfnm{Weili}\binits{W.}} \AND
  \bauthor{\bsnm{Marron},~\bfnm{JS}\binits{J.}}
(\byear{2012}).
\btitle{Local polynomial regression for symmetric positive definite matrices}.
\bjournal{Journal of the Royal Statistical Society: Series B (Statistical
  Methodology)}
\bvolume{74}
\bpages{697--719}.
\end{barticle}
\endbibitem

\bibitem[\protect\citeauthoryear{Ziezold}{1977}]{ziez:77}
\begin{binproceedings}[author]
\bauthor{\bsnm{Ziezold},~\bfnm{Herbert}\binits{H.}}
(\byear{1977}).
\btitle{On expected figures and a strong law of large numbers for random
  elements in quasi-metric spaces}.
In \bbooktitle{Transactions of the Seventh Prague Conference on Information
  Theory, Statistical Decision Functions, Random Processes and of the 1974
  European Meeting of Statisticians}
\bpages{591--602}.
\bpublisher{Springer}.
\end{binproceedings}
\endbibitem

\end{thebibliography}
\end{document}